\documentclass[12pt, a4paper,reqno]{amsart}

\usepackage[T1]{fontenc}
\usepackage[utf8]{inputenc}

\usepackage{color}
\definecolor{bleu_sombre}{rgb}{0,0,0.6}  \definecolor{rouge_sombre}{rgb}{0.8,0,0}\definecolor{vert_sombre}{rgb}{0,0.6,0}
\usepackage[plainpages=false,colorlinks,linkcolor=bleu_sombre,
citecolor=rouge_sombre,urlcolor=vert_sombre,breaklinks,unicode]{hyperref}

\usepackage[english]{babel}
\usepackage{geometry}

\usepackage{amsmath,amssymb,amsthm,graphicx,amsfonts,url,enumerate,dsfont,stmaryrd, mathrsfs, csquotes,braket,url}

\theoremstyle{plain}
\newtheorem{theorem}{{Theorem}}[section]
\newtheorem*{theorem*}{{Theorem}}
\newtheorem{proposition}[theorem]{Proposition}

\newtheorem*{proposition*}{Proposition}
\newtheorem{corollary}[theorem]{Corollary}
\newtheorem*{corollary*}{Corollary}
\newtheorem{lemma}[theorem]{Lemma}
\newtheorem{assumption}[theorem]{Assumption}
\newtheorem*{lemma*}{Lemma}
\theoremstyle{definition}
\newtheorem{definition}[theorem]{Definition}
\newtheorem*{definition*}{Definition}
\theoremstyle{remark}
\newtheorem{remark}[theorem]{Remark}
\newtheorem{notation}[theorem]{Notation}

\makeatletter

\@addtoreset{equation}{section}
\makeatother

\renewcommand{\leq}{\leqslant}	\renewcommand{\geq}{\geqslant}

\newcommand{\R}{\mathbb{R}}	
\newcommand{\C}{\mathbb{C}}
\newcommand{\N}{\mathbb{N}}	

\newcommand{\dd}{\mathrm{d}}

\renewcommand{\Re}{\mathrm{Re}\,}

\begin{document}

\newcommand\red{\textcolor{red}}

\title[]{Magnetic Dirac operator in strips\\ submitted to strong magnetic fields}

\author[L. Le Treust]{Loïc Le Treust}
\address[L. Le Treust]{Aix Marseille Univ, CNRS, I2M, Marseille, France}
\email{loic.le-treust@univ-amu.fr}

\author[N. Raymond]{Nicolas Raymond}
\address[N. Raymond]{Univ Angers, CNRS, LAREMA, Institut Universitaire de France, SFR MATHSTIC, F-49000 Angers, France}
\email{nicolas.raymond@univ-angers.fr}

\author[J. Royer]{Julien Royer}
\address[J. Royer]{Institut de math\'ematiques de Toulouse, Universit\'e Toulouse III, 118 route de Narbonne, F-31062 Toulouse cedex 9, France}
\email{julien.royer@math.univ-toulouse.fr}

\begin{abstract}
We consider the magnetic Dirac operator on a curved strip whose boundary carries the infinite mass boundary condition. When the magnetic field is large, we provide the reader with accurate estimates of the essential and discrete spectra. In particular, we give sufficient conditions ensuring that the discrete spectrum is non-empty.
	
\end{abstract}
\maketitle
	
\tableofcontents{}

\section{Motivations and main results}

\subsection{The magnetic Dirac operator on a strip}
We perform the spectral analysis of magnetic Dirac operators on bidimensional strips. 

\begin{figure}[htb]
    \centering
    \includegraphics[width=1\textwidth]{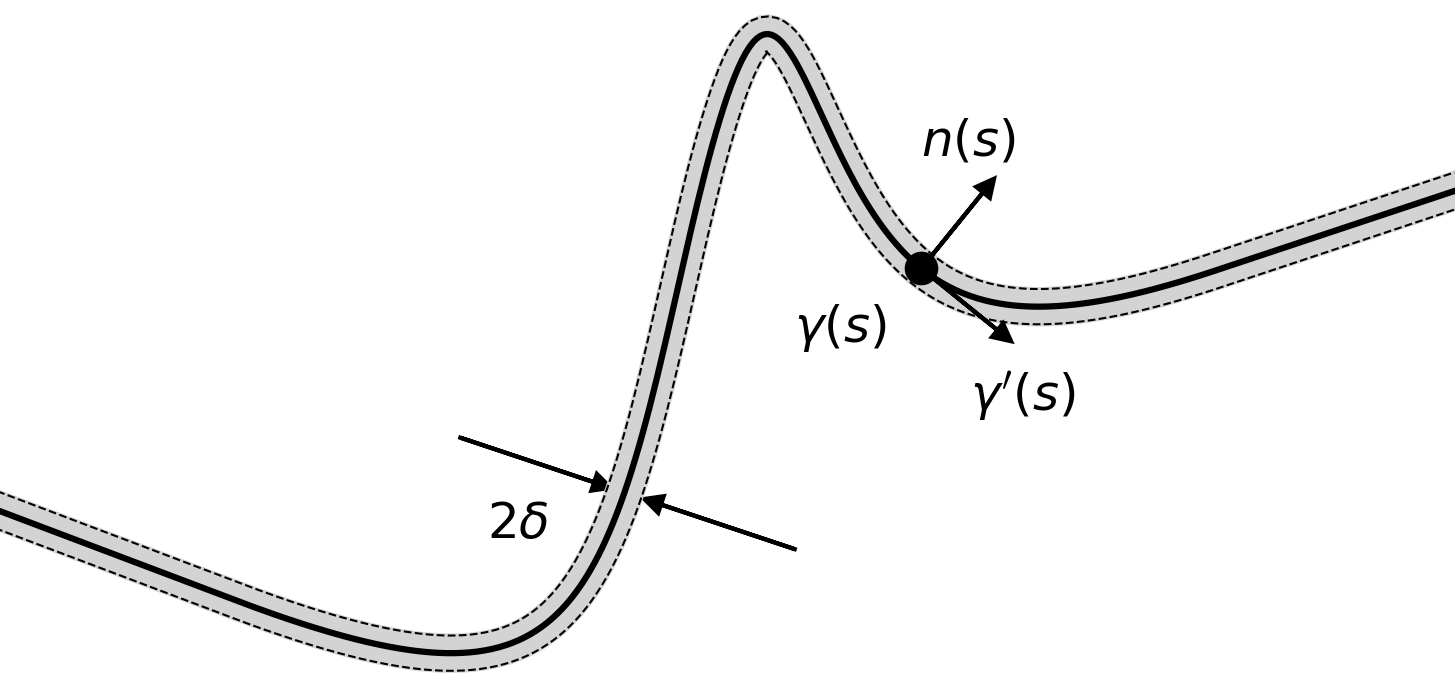}
    \caption{Typical Waveguide Profile}
    \label{fig:waveguide}
\end{figure}

\subsubsection{The strips}
The strips under consideration in this article are built from a smooth curve without self-intersections $\gamma : \R\to \R^2$ (with $|\gamma'|=1$) and from the application
\begin{equation}\label{eq:Diffeo}
	\Theta : \Omega_0\ni(s,t)\mapsto \gamma(s)+t\mathbf{n}(s)\,,
\end{equation}
where $\mathbf{n}=\gamma'^{\perp}$ is chosen so that $(\gamma',\bf{n})$ forms a direct orthonormal basis and $\Omega_0=\R\times(-\delta,\delta)$ is a straight strip of width $\delta>0$ so small that $\Theta$ is injective. The curvature \( \kappa \) of $\gamma$ is characterized by
\[
	\gamma^{\prime\prime}(s) = \kappa(s)\textbf{n}(s)
\]
for all \(s \in \mathbb{R}\). To simplify the analysis, we only consider the case when the curvature has compact support. Then, the strip is $\Omega=\Theta(\Omega_0)$, which, for $\delta$ small enough, is a smooth curved strip about the base curve $\gamma$ (see Figure \ref{fig:waveguide}).

\subsubsection{The choice of magnetic gauge}
In order to define the Dirac operator with constant magnetic field $B=1$, we need an associated vector potential $\mathbf{A} : \overline{\Omega}\to\R^2$ (that is a function such that $\partial_1A_2-\partial_2A_1=1$). Note that any two associated vector potentials yield unitarily equivalent operators, given that $\Omega$ is simply connected. In such a geometric context, there is a rather natural choice, whose nice properties lighten the presentation. 

For the straight strip, consider the bounded function $\phi_0$ on $\Omega_0$ given by 
\[
\phi_0(s,t)=\frac{t^2-\delta^2}2\,,
\]
which satisfies $\mathrm{curl}\mathbf{A}_0=1$, with  $\mathbf{A}_0=\nabla\phi^\perp_0=(-t,0)$. In particular, note that $\mathbf{A}_0$ is bounded. To get a smooth and bounded magnetic vector potentiel for the curved strip $\Omega$, we use the following proposition, established in \cite[Proposition 1.2]{BLLTRR23}.
\begin{proposition}\label{prop.phi}
	Let $\hat \phi_0 = \phi_0 \circ \Theta^{-1} \in \mathscr C^\infty(\overline \Omega)$. There exists a unique $\phi\in\mathscr{C}^\infty(\overline{\Omega})$ such that $\Delta\phi=1$, $\phi_{|\partial\Omega}=0$, and $\phi-\hat \phi_0\in\mathscr{S}(\overline{\Omega})$. Moreover, there exists $c_0>0$ such that $\partial_{\mathbf{N}} \phi \geq c_0$ on $\partial\Omega$, $\mathbf{N}$ being the outward pointing normal to the boundary.
\end{proposition}
Thanks to the function $\phi$ given in Proposition \ref{prop.phi}, we get the existence of a smooth and bounded vector potential $\mathbf{A}=\nabla\phi^\perp$ on $\Omega$.

\subsubsection{The magnetic Dirac operators}
For $h>0$, we consider the magnetic Dirac operators
\[\mathscr{D}_{h}=\sigma\cdot(\mathbf{p}-\mathbf{A})
=\begin{pmatrix}
0&d_{h,\mathbf{A}}\\
d^\times_{h,\mathbf{A}}&0
\end{pmatrix}
=\begin{pmatrix}
0&-2ih\partial_{z}-A_1+iA_2\\
-2ih\partial_{\overline{z}}-A_1-iA_2&0
\end{pmatrix}\,,
\]
where $\mathbf{p}=-ih\nabla$, $\partial_z = \frac{\partial_1 -i\partial_2}{2}$, $\partial_{\overline{z}} = \frac{\partial_1 +i\partial_2}{2}$, and the Pauli matrices are given by
\[
\sigma_1 = \left(
\begin{array}{cc}
	0&1\\1&0
\end{array}
\right),
\quad
\sigma_2 = \left(
\begin{array}{cc}
	0&-i\\i&0
\end{array}
\right),
\quad
\sigma_3 = \left(
\begin{array}{cc}
	1&0\\0&-1
\end{array}
\right)\,,
\]
and \[\mathscr{D}_{h,0}=\sigma\cdot(\mathbf{p}-\mathbf{A}_0)=\begin{pmatrix}
0&-ih\partial_{s}-h\partial_{t}+t\\
-ih\partial_{s}+h\partial_{t}+t&0
\end{pmatrix}\,,\] 
 with respective domains
\[\begin{split}
\mathrm{Dom}(\mathscr{D}_h)&=\{\psi\in H^1(\Omega,\mathbb{C}^2) : -i\sigma_3(\sigma\cdot\mathbf{N})\psi=\psi\,,\text{ on }\partial\Omega\}\,,\\
\mathrm{Dom}(\mathscr{D}_{h,0})&=\{\psi\in H^1(\Omega_0,\mathbb{C}^2) : -i\sigma_3(\sigma\cdot\mathbf{N})\psi=\psi\,,\text{ on }\partial\Omega_0\}\,,
\end{split}\]
where $\mathbf{N}$ is the outward pointing normal to the boundary $\partial\Omega$ ($\mathbf{N} = \pm\mathbf{n}$). The boundary conditions are the so-called infinite mass boundary conditions. 

\begin{remark}
To establish that $\mathscr{D}_{h,0}$ and $\mathscr{D}_{h}$ are well-defined on $H^1$, and that they are self-adjoint, we rely on the boundedness of $\mathbf{A}_0$ and $\mathbf{A}$, as well as on the properties of the nonmagnetic Dirac operators and the analysis presented in \cite{ALTR20}. Introducing a gauge-equivalent unbounded vector potential would complicate the definition of the domain and the proof of self-adjointness. Moreover, the boundedness of these potentials is not the only reason for our choice of gauge. The functions $\phi_0$ and $\phi$ play a central role in the main results of this paper, as they encode the ‘magnetic geometry’ of the problem.
\end{remark}

\begin{remark}
	Note that the boundary condition can also be written as $\psi_2=\pm in\psi_1$ where $n=n_1+in_2$, $(n_1,n_2)$ being the coordinates of $\mathbf{n}$. This boundary condition, commonly known as the infinite mass or MIT bag model boundary condition, originates from the study of relativistic quantum particles. It was first introduced in 3D by Bogoliubov et al. \cite{Bogoliubov1968} in the context of the Bogoliubov bag model, and later extended by Chodos et al. in the MIT bag model to describe quark confinement \cite{MITBagModel1974}. In 2D, Berry and Mondragon \cite{Berry1987} applied the infinite mass boundary condition to neutrino billiards, marking its introduction in two dimensions. They also explored other local boundary conditions. The infinite mass boundary condition has since been adopted in the study of graphene quantum dots and nanoribbons, where it ensures the confinement of electrons and the presence of a spectral gap, in contrast to other boundary conditions such as zigzag, which allow zero-energy states \cite{CastroNeto2009}. We refer to the introduction of \cite{BLTRS23} for an extended description of the various boundary conditions, and to \cite[Section 1.3]{BLTRS23} for a detailed exploration of the relationship between Dirac operators with zigzag boundary conditions and the square root of the Dirichlet Pauli operators. 
\end{remark}

\subsection{Main results}
The aim of the article is to study the spectra of the operators $\mathscr{D}_{h,0}$ and $\mathscr{D}_{h}$ in the limit $h\to 0$ (which is equivalent to the large magnetic field limit, with a magnetic field of strength $h^{-1}$). Our first result describes their essential spectra by providing the reader with asymptotic estimates of the negative and positive thresholds of the essential spectrum.
\begin{theorem}\label{thm.ess}
For all $h>0$,
\[\mathrm{sp}_{\mathrm{ess}}(\mathscr{D}_h)=\mathrm{sp}(\mathscr{D}_{h,0})=\mathrm{sp}_{\mathrm{ess}}(\mathscr{D}_{h,0})\,.\]	
Moreover, for all $h>0$, there exist $\lambda^{\pm}_{\mathrm{ess}}(h)>0$ such that:
\[\mathrm{sp}(\mathscr{D}_{h,0})=\mathbb{R}\setminus(-\lambda^-_{\mathrm{ess}}(h),\lambda^+_{\mathrm{ess}}(h))\,,\]	
and we have
\[\lambda^+_{\mathrm{ess}}(h)=2\sqrt{\frac{h}{\pi}}e^{-\delta^2/h}(1+o(1)) \quad\textrm{and} \quad \lambda_{\mathrm{ess}}^-(h) = a_0 \sqrt{h} +\mathscr{O}(h^\infty)\,,\]
for some $a_0\in(0,\sqrt{2})$.
\end{theorem}
As expected, the essential spectra of $\mathscr{D}_h$ and $\mathscr{D}_{h,0}$ coincide, since $\Omega$ looks like $\Omega_0$ at infinity. The constant $a_0$ is the one appearing in \cite[Theorem 1.15]{BLTRS23}: it represents the spectral gap of the Dirac operator with magnetic field equal to $1$ when $h=1$ on a half-plane. 
\begin{remark}\label{rem.hypcurvature}
Note that one could relax our assumption that the curvature has compact support by assuming that $\kappa$ goes to $0$ at infinity sufficiently fast.
\end{remark}
Let us now discuss the existence of the discrete spectrum for $\mathscr{D}_h$. To ensure its existence, we will work under the following assumption.
\begin{assumption}\label{assu.0}
The function $\phi$ has a unique minimum attained at $x_{\min} (\in \Omega)$, which is non-degenerate. Moreover, we have $\phi_{\min}=\min_{\overline{\Omega}}\phi<\min_{\overline{\Omega}_0}\phi_{0}=-\frac{\delta^2}{2}$ and $\liminf_{\substack{|x| \to \infty \\ x \in \Omega }} \phi(x) = \min_{\overline{\Omega}_0}\phi_{0}> \phi_{\min}$.
\end{assumption}
It is known from \cite[Proposition 1.3]{BLLTRR23} and from Proposition \ref{prop.phi} above (\cite[Proposition 1.2]{BLLTRR23}) that Assumption \ref{assu.0} is satisfied when the strip is straight away from a compact set, thin enough, and when the square of the curvature $\kappa$ of its base curve $\gamma$ has a unique maximum, which is non-degenerate.

\begin{remark}
In this paper, we do not consider the thin waveguide limit $\delta \to 0$. Instead, we fix $\delta > 0$, assuming that Assumption \ref{assu.0} holds. The constants in our estimates may depend on $\delta$.
 \end{remark}

In order to formulate our main theorem, one will need the Segal-Bargmann space
	\[\mathscr{B}^2(\C) = \{u\in\mathscr{O}(\mathbb{C}) :N_{\mathscr{B}}(u)<+\infty\}\,,\] 
where
	\[
N_{\mathscr{B}}(u)=\left(\int_{\mathbb{R}^2} \left|u \left(y_1+iy_2\right)\right|^2e^{-\mathsf{Hess}_{x_{\min}}\phi(y, y)} \dd y\right)^{1/2}\,.
\]
One will also need the Hardy space $\mathscr{H}^2(\Omega)$ on $\Omega$, which is essentially made of holomorphic functions on $\Omega$ having a trace on $\partial\Omega$ that is $L^2(\partial\Omega)$. The Hardy space is equipped with $\|\cdot\|_{L^2(\partial\Omega)}$. More details about the Hardy space and its norm are given in Appendix \ref{app.A}, see also the discussion in Section \ref{sec.strategy}. The distances associated with the above norms are denoted by $\mathrm{dist}_{\mathscr{B}}$ and $\mathrm{dist}_{\mathscr{H}}$.

For $k \in \N^*$ we set (we will write $z_{\min}$ instead of $x_{\min}$ when $\Omega$ is considered as a subset of $\C$)
\[\begin{split}
	\mathbb{X}_k &= \{u\in \mathscr{H}^2(\Omega)\,, \forall j\in\{0,\dots,k-2\}\,,\;u^{(j)}(z_{\rm min})=0\,,\; u^{(k-1)}(z_{\rm min})=1\}\,,
	\\
	\mathbb{Y}_k&=\{u\in\mathbb{C}[X]\,,\deg u = k-1\,, u^{(k-1)}(0) = 1\}\,.
\end{split}\]
Then we set $d^{k}_\mathscr{H} = {\rm dist}_{\mathscr{H}^2(\Omega)}(0,\mathbb{X}_k)$ and $d^{k}_\mathscr{B} = {\rm dist}_{\mathscr{B}(\mathbb{C})}(0,\mathbb{Y}_k)$.

Here comes our main theorem.
\begin{theorem}\label{thm.boundstatesDirac}
Suppose that Assumption \ref{assu.0} holds.
\begin{enumerate}[\rm (i)]
\item\label{it.AsymEff+}
Let $k\geq 1$. Consider
\[\lambda^{\mathrm{eff}}_k(h)=\inf_{\underset{\dim W=k}{W\subset \mathscr{H}^2(\Omega)}}\sup_{u\in W\setminus\{0\}}\frac{h\|u\|^2_{\partial\Omega}}{\|e^{-\phi/h}u\|^2}\,,\]
Then, we have
\[
\lambda_k^{\rm eff}(h) = h^{1-k}e^{2\phi_{\rm min}/h}\left(\frac{d^{k}_\mathscr{H}}{d^{k}_\mathscr{B}}\right)^2(1+o_{h\to 0}(1))\,.
\]
\item \label{it.AsymEff+2}Consider $N\in\mathbb{N}^*$. There exists $h_0>0$ such that for all $h\in(0,h_0)$ the operator $\mathscr{D}_h$ has at least $N$ positive discrete eigenvalues (counted with multiplicities). Denoting the first $N$ eigenvalues by $(\lambda^+_k(h))_{k\in\{1,\ldots,N\}}$, we have for all $k\in\{1,\ldots,N\}$
\[\lambda^+_k(h) \mathop{\mbox{$\sim$}} \limits_{h \to 0}  \lambda^{\mathrm{eff}}_k(h)\,.\]

\end{enumerate}
\end{theorem}
\begin{remark}
~
	\begin{enumerate}[\rm (i)]
		\item  Theorem \ref{thm.boundstatesDirac} establishes the non-emptyness of the discrete spectrum when $h$ is small enough. A similar question has recently been considered for the Dirichlet-Pauli operator in \cite{BLLTRR23} with some of the ideas from \cite{BLTRS21}: it solved an open problem by P. Duclos and P. Exner (see \cite{Exner12} and the non-exhaustive literature \cite{DE95, OM05, KR14, EK15} about waveguides, sometimes with magnetic fields). In the present article, we also provide the reader with the one-term asymptotics of the smallest positive eigenvalues and not only upper bounds.
		\item Theorem \ref{thm.boundstatesDirac} is an extension of \cite[Theorem 1.12]{BLTRS23} to unbounded and non-convex domains. We underline that some of the geometric quantities are related to the Hardy space on the curved strip $\mathscr{H}^2(\Omega)$ and that the polynomials do not belong to this space, contrary to the case when $\Omega$ is bounded. This is the reason why the constants attached to the Hardy space are written in a way slightly different from \cite[Theorem 1.12]{BLTRS23}. This absence of the polynomials in the Hardy space has important consequences on the proof, see Section \ref{sec.strategy} below.

\item The case of magnetic Dirac operators on annuli has been considered in \cite{lavignebon:hal-04037383,lavigne:hal-03275306}.
\item
When $\gamma$ is analytic and $\kappa$ tends to zero at infinity, we can show that Theorem \ref{thm.ess} remains valid, ensuring that the essential spectrum of the Dirac operator coincides with that of the Dirac operator on the straight strip (see Remark \ref{rem.hypcurvature}). Then, similarly to \cite[Theorem 1.22]{BLTRS23}, where analyticity plays a key role, we can prove that the smallest (in absolute value) negative eigenvalue, denoted by $-\lambda_1^-(h)<0$,  lies inside the spectral gap when $h$ is sufficiently small.

  Moreover, for some $\mathfrak{c}_0>0$, we have
\[\lambda_1^-(h)=a_0\sqrt{h}+h^{\frac32}\mathfrak{c}_0\lambda+o(h^{\frac32})\,,\] 
where the groundstate energy satisfies \[\lambda = \min\left(\lambda_1\left(D_s^2-\frac{\kappa(s)^2}{12(1-\delta\kappa)^2}\right),\lambda_1
\left(
D_s^2-
\frac{\kappa(s)^2}
{12(1+\delta\kappa)^2}\right)\right)<0\,.\]
\item Non-magnetic Dirac operators on tubes with infinite mass boundary condition have recently been considered in \cite{BBKO22, BKO22, K23, LTOBR24}, especially in the thin width limit. Note that this limit is not considered in the present paper, where the domain is fixed and the constant $h$ tends to $0$. The asymptotic problems we encounter are fully bidimensional (at least for positive eigenvalues). This contrasts with the thin tube limit, where the scales at which the normal and tangential operators naturally act differ, leading to simpler one-dimensional models.
\end{enumerate}
\end{remark}

\begin{remark}
	There are explicit expressions for the constants $(d^k_\mathscr{B})_k$ and $(d^k_\mathscr{H})_k$.
	\begin{enumerate}[\rm (i)]
	\item The sequence of $N_{\mathscr{B}}$-orthogonal monic polynomials $(P_m)_{m\geq 0}$ with
\[
	P_m\colon z\mapsto z^m + \sum_{n=0}^{m-1}c_{m,n}z^n\,,m\in\mathbb{N}\,,c_{m,n}\in\mathbb{C}\,,
\]
obtained after a Gram-Schmidt process over the family $(z^k)_{k\geq0}$ satisfies
for $m\geq 0$, $(d^{m+1}_\mathscr{B}) =N_{\mathscr{B}}(P_{m})/m!$,
		\[
			P_m\colon z\mapsto \begin{cases}
				z^m&\text{ if }a=b\,,
				\\
				\left|\frac{b-a}{ab}\right|^{m/2}He_m
					\left(
						z\sqrt{\left|\frac{ab}{b-a}\right|}
					\right)&\text{ if }a\ne b\,,
			\end{cases}
		\]
and
	\[
		N_{\mathscr{B}}(P_m)^2 = \frac{2\pi m! (a + b)^m}{(ab)^{m+\frac{1}{2}}}\,,
	\]
where $a/2,b/2$ are the eigenvalues of $\mathsf{Hess}_{x_{\min}}\phi$ and $He_m$ are the probabilist's Hermite polynomials \cite[\href{https://dlmf.nist.gov/18.3}{Section 18.3}]{NIST:DLMF},
\[
	He_m\colon z\mapsto (-1)^{m}\mathrm {e} ^{z^{2}/2}{\frac {\mathrm {d} ^{m}}{\mathrm {d} z^{m}}}\mathrm {e} ^{-z^{2}/2}
	 = \sum_{l=0}^{\left\lfloor m/2 \right\rfloor}\frac{(-1)^lm!}{2^ll!(m-2l)!}z^{m-2l}\,.
\]
Therefore, we have
\[
(d^k_\mathscr{B})^2 =N_{\mathscr{B}}(P_{k-1})^2/[(k-1)!]^2
 =
  \frac{\pi (B(x_{\rm min}))^{k-1}}{2^{k-1}(k-1)!(\det\,\mathsf{Hess}_{x_{\min}}\phi)^{k-\frac{1}{2}}}\,,k\geq 1
 \,.
\]
The isotropic case $a=b$ is straightforward : $(P_n)_{n\geq 0} = (z^n)_{n\geq 0}$. The anisotropic case is a consequence of \cite{MR1041203} with $A = \min \left(\frac{a}{b}, \frac{b}{a}\right)$ and a change of scale $\tilde z = \left|\frac{b-a}{ab}\right| z$ (see also \cite{MR3935204}). For a general presentation on orthogonal polynomials, see \cite[Section 2.3.4]{MR3289583}.	
	\item For $k\geq 1$,
	\[
		d^k_{\mathscr{B}} = 
			\inf\left\{
				\frac{\|u\|_{\mathscr{H}^2(\Omega)}}{|u^{(k-1)}(z_{\rm min})|}
				\,,
				u\in \mathscr{H}^2(\Omega)\,, 
				\begin{array}{l}
					u^{(j)}(z_{\rm min}) =0\,, \forall j\in\{0,\dots,k-2\}\,,\\
					u^{(k-1)}(z_{\rm min})\ne0
				\end{array}
				\right\}\,.
	\]
	On the unit disk $\mathbb{D}$, the sequence $(z^{k-1})_{k\geq 1}$ realizes the minima and their values are $\left(\frac{\sqrt{2\pi}}{(k-1)!}\right)_{k\geq 1}$. Notice then that the sequence $\left(\Lambda z^{k-1}\right)_{k\geq 1}$ realizes the minima on $\Omega$ where $\Lambda$ is the isometric isomorphism defined by
	\[
		\begin{array}{llll}
			\Lambda\colon &\mathscr{H}^2(\mathbb{D})&\longrightarrow &\mathscr{H}^2(\Omega)
			\\
			&u&\longmapsto&\left[z\mapsto \sqrt{\varphi'(z)}u\circ\varphi(z)\right]\,,
		\end{array}
	\]
	$\varphi$ being a biholomorphism from the unit disk $\mathbb{D}$ to $\Omega$ such that $\varphi(0) = z_{\rm min}$.
	Therefore, we have
	\[
	d^k_\mathscr{H} = \frac{\sqrt{2\pi}}{(k-1)!}|\varphi'(0)|^{k-\frac{1}{2}}\,.
	\]
	\item We have for $k\geq 1$,
	\[
	\left(\frac{d^k_\mathscr{H}}{d^k_\mathscr{B}}\right)^2
	=
	\frac{2^k|\varphi'(0)|^{2k-1}(\det\,\mathsf{Hess}_{x_{\min}}\phi)^{k-\frac{1}{2}}}{  (k-1)!B(x_{\rm min})^{k-1}}.
	\]
	
	\end{enumerate}
\end{remark}

\subsection{Organization and strategy}\label{sec.strategy}
Section \ref{sec.2} is devoted to the proof of Theorem \ref{thm.ess}. In Section \ref{sec.21}, we show that $\mathscr{D}_h$ and $\mathscr{D}_{h,0}$ have the same essential spectrum, see Proposition \ref{prop.esscoin}. To do so, we prove that $\mathscr{D}_h$ is unitarily equivalent to an operator on the straight strip $\Omega_0$ and we use the Weyl criterion. In Section \ref{sec.22}, we study the spectrum of $\mathscr{D}_{h,0}$ be means of the Fourier transform in the longitudinal variable. We get a family of one dimensional Dirac operators $(\mathscr{D}_{h,0,\xi})_{\xi\in\R}$ (which have compact resolvent). For each $\xi$, the spectrum of $\mathscr{D}_{h,0,\xi}$ is made of positive eigenvalues $(\mu_n^+(\xi,h))_{n\geq 1}$ and of negative eigenvalues $(-\mu_n^-(\xi,h))_{n\geq 1}$, which are even function of $\xi$ (see Lemma \ref{lem.even}). Then, we focus on a description of $\mu_1^+(\xi,h)$, which is characterized in Proposition \ref{prop.charact}. This characterization implies an estimate of $\inf_{\xi\in\R}\mu_1^+(\xi,h)$, see Proposition \ref{prop.curves+}, and of the threshold $\lambda^+_{\rm ess}(h) =\inf_{\xi\in\mathbb{R}}\mu_1^+(\xi,h)$, see Corollary \ref{cor.lambdaess+s}. Section \ref{sec.25} is devoted to the estimate of $\inf_{\xi\in\R}\mu_1^-(\xi,h)=\lambda^-_{\mathrm{ess}}(h)$. In Section \ref{sec.3}, we prove Theorem \ref{thm.boundstatesDirac}. Sections \ref{sec.31} (upper bound) and \ref{sec.32} (lower bound) establish Point \eqref{it.AsymEff+}. We emphasize that the polynomials do not belong to the Hardy space on $\Omega$ and that Taylor expansions near $x_{\min}$ have to be replaced by a suitable "Taylor expansion" in the Hardy space $\mathrm{Tayl}_{\mathscr{H}^2(\Omega)}$. More precisely, one has to approximate functions by functions in the Hardy space having the same Taylor expansion at $x_{\min}$, see Notation \ref{not.quasiMode2} and Lemma \ref{lem.tayl2}. Up to this key idea (which actually allows to deal with general unbounded domains), the proof follows then the same steps as in \cite[Section 3]{BLTRS23}. Section \ref{sec.33} is devoted to the proof of Point \eqref{it.AsymEff+2}. We start by introducing some $\mu_k(h)$ in \eqref{eq.mukh}. These numbers will turn to be exactly the $\lambda^+_k(h)$. To check that, one first must check that they do not belong to the essential spectrum when $h$ is small enough, see Proposition \ref{prop.asymptoticmu}. The analysis of the Fredholmness is the key to deal with the fact that $\mathscr{D}_h$ does not have compact resolvent. A crucial point that allows the connection between the $\mu_k(h)$ and the $\lambda^+_k(h)$ is Proposition \ref{prop.J} \eqref{eq.Jiv}, see Section \ref{sec.proof312}.

\section{Estimate of the essential spectrum}\label{sec.2}

\subsection{\texorpdfstring{The operators $\mathscr{D}_h$ and $\mathscr{D}_{h,0}$ share the same essential spectrum}{The operators on the straight and curved waveguides share the same essential spectrum}}\label{sec.21}

For $(s,t) \in \Omega_0$ we set
\[
m(s,t) = 1 - t \kappa(s).
\]
Then we consider on $L^2(\Omega_0;\C)$ the operator $\mathfrak{D}_h$ defined by
\[
\mathfrak{D}_h =
			\frac{\sigma_1}{2}\left(m^{-1}\left(hD_s + t-\frac{t^2\kappa}{2}\right) + \left(hD_s + t-\frac{t^2\kappa}{2}\right)m^{-1}\right)
		+\sigma_2(hD_t)\,,
\]
on the domain
\[
\mathrm{Dom}(\mathfrak{D}_h) = \{\varphi \in H^1(\Omega_0;\C^2) : \varphi_2(s,\pm \delta) = \mp \varphi_1(s,\pm \delta), \forall s \in \R\}.
\]

The following lemma follows from standard arguments (see for instance \cite[Theorem 2.1]{LTOBR24}).
\begin{lemma}\label{lem.Dhunitary}
	The operator $\mathscr{D}_h$ is unitarily equivalent to	$\mathfrak{D}_h$.
\end{lemma}
\begin{proof}
	Let us describe the action of $\mathscr{D}_h$ in the tubular coordinates $(s,t)$ given by $x=\Theta(s,t)=\gamma(s)+t\mathbf{n}(s)$. Since $\textbf n = \gamma'^\perp$ and $\gamma'' = \kappa \textbf{n}$ we can write ${\rm Jac}\,\Theta(s,t)^{-T} = \begin{pmatrix}m^{-1}\gamma ' (s)\;, \textbf{n}(s)\end{pmatrix}$, and by the chain rule,
	\[
		\Theta^*\nabla_{x}(\Theta^*)^{-1} = \gamma'm^{-1}\partial_s + \textbf{n}\partial_t\,.
	\]
	Let us also consider the new vector potential 
	\[
		\widetilde{\mathbf{A}} = {\rm Jac}\,\Theta(s,t)^{T}(\mathbf{A}\circ\Theta) = \begin{pmatrix}m(A\circ \Theta)\cdot \gamma'\\(A\circ \Theta)\cdot \bf{n}\end{pmatrix}\,,
	\]
	 which satisfies $\mathrm{curl}\,\widetilde{\mathbf{A}}=1-t\kappa(s)$. We obtain
	 \[
	 \Theta^*\mathscr{D}_h (\Theta^*)^{-1} =\Theta^* (\sigma\cdot (\textbf{p}-\textbf{A})) (\Theta^*)^{-1}
	 	= m^{-1}(\sigma\cdot\gamma')(-ih\partial_s - \widetilde A_s)
		+(\sigma\cdot\textbf{n})(-ih\partial_t - \widetilde A_t).
	 \]
This is an equality between unbounded operators on the weighted space $L^2(\Omega_0,m\mathrm{d}s\mathrm{d}t)$. 
After conjugaison by $m^{1/2}$, we obtain that $\mathscr D_h$ is unitarily equivalent to the following operator on $L^2(\Omega_0,m\mathrm{d}s\mathrm{d}t)$.
	 \[\begin{split}
		&m^{-1}(\sigma\cdot\gamma')(-im^{1/2}h\partial_sm^{-1/2} - \widetilde A_s)
		+(\sigma\cdot\textbf{n})(-im^{1/2}h\partial_tm^{-1/2}  - \widetilde A_t)\,,
		\\&=
		m^{-1}(\sigma\cdot\gamma')(-ih\partial_s - \widetilde A_s - \frac{ith\kappa'}{2m})
		+(\sigma\cdot\textbf{n})(-ih\partial_t  - \widetilde A_t - \frac{ih\kappa}{2m})\,.
	 \end{split}\] 
	Let us recall that the Dirac equation is covariant. In particular, if $\gamma'(s) = e^{i\theta}$, then
	\[
		e^{i\sigma_3\theta/2}\sigma\cdot \gamma'e^{-i\sigma_3\theta/2} = \sigma_1\,,
		\quad
		e^{i\sigma_3\theta/2}\sigma\cdot \textbf{n}e^{-i\sigma_3\theta/2} = \sigma_2\,,
	\]
	and
	\[
		e^{i\sigma_3\theta/2}(-ih\partial_s)e^{-i\sigma_3\theta/2}
		=
		(-ih\partial_s) 
		-
		\frac{h\kappa \sigma_3}{2}\,,
	\]
	so that a conjugaison by the rotor $e^{-i\sigma_3\theta/2}$ leads to
		 \[
			m^{-1}\sigma_1(-ih\partial_s - \widetilde A_s - \frac{iht\kappa'}{2m})
		+\sigma_2(-ih\partial_t  - \widetilde A_t)
	 \]
	 which can be rewritten in an explicitly symmetric form as
	\[
			\frac{\sigma_1}{2}\left(m^{-1}(-ih\partial_s - \widetilde A_s) + (-ih\partial_s - \widetilde A_s)m^{-1}\right)
		+\sigma_2(-ih\partial_t  - \widetilde A_t)\,.
	 \]	 
		Note also that by the covariance of the Dirac operator, this operator is equipped with the infinite mass boundary condition on $\Omega_0$. Indeed, we have 
		\[
			e^{i\sigma_3\theta/2}\left(-i\sigma_3\sigma\cdot \textbf{n}\right)e^{ - i\sigma_3\theta/2} = -i\sigma_3\sigma_2 = \sigma_1
		\]
		and
		$\varphi_2(s,\pm\delta)=\mp\varphi_1(s,\pm\delta)$.
	Finally, we have
	\[
		{\rm curl }\begin{pmatrix}
			-t+\frac{t^2\kappa}{2}\\0
		\end{pmatrix} = 1-t\kappa = m = {\rm curl }\tilde A\,,
	\]
	and that $\Omega_0$ is simply connected. Hence, there exists a change of gauge so that $\mathscr{D}_h$ is unitarily equivalent to
	\[
			\frac{\sigma_1}{2}\left(m^{-1}\left(-ih\partial_s + t-\frac{t^2\kappa}{2}\right) + \left(-ih\partial_s + t-\frac{t^2\kappa}{2}\right)m^{-1}\right)
		+\sigma_2(-ih\partial_t)\,.
	 \]

\end{proof}

\begin{proposition}\label{prop.esscoin}
	We have
	\[\mathrm{sp}_{\mathrm{ess}}(\mathscr{D}_h)=\mathrm{sp}_{\mathrm{ess}}(\mathscr{D}_{h,0})\,.\]
\end{proposition}

\begin{proof}
	Thanks to Lemma \ref{lem.Dhunitary}, we may focus on $\mathfrak{D}_h$. We have
	\[\mathfrak{D}_h=\mathfrak{D}_{h,0}+\mathfrak{P}_h\,,\]
	where $\mathfrak{D}_{h,0} = \mathscr{D}_{h,0} = \sigma_1(hD_s + t)+\sigma_2D_t$ and $\mathfrak{P}_h=V_h+\frac{1}{2}\left(W_h hD_s + hD_sW_h\right)$
	with
	\[
		V_h = \sigma_1\left((m^{-1}-1)t - \frac{t^2\kappa}{2m}\right)
		\,,\quad
		W_h = \sigma_1(m^{-1}-1)\,.
	\]
	Since $\kappa$ is compactly supported, so are $V_h$ and $W_h$.
	Let us explain why $(\mathfrak{D}_h+i)^{-1}-(\mathfrak{D}_{h,0}+i)^{-1}$ is compact. In virtue of the Weyl criterion, this will imply that $\mathrm{sp}_{\mathrm{ess}}(\mathfrak{D}_h)=\mathrm{sp}_{\mathrm{ess}}(\mathfrak{D}_{h,0})$ and thus the conclusion.
	
	We use the resolvent formula to get
	\[\begin{split}(\mathfrak{D}_h+i)^{-1}-(\mathfrak{D}_{h,0}+i)^{-1}&=(\mathfrak{D}_{h,0}+i)^{-1}(\mathfrak{D}_{h,0}-\mathfrak{D}_{h})(\mathfrak{D}_{h}+i)^{-1}\\
		&=-(\mathfrak{D}_{h,0}+i)^{-1}\mathfrak{P}_h(\mathfrak{D}_{h}+i)^{-1}\,.
	\end{split}\]
	Since $\mathfrak{D}_h$ and $\mathfrak{D}_{h,0}$ are self-adjoint, $hD_s(\mathfrak{D}_{h}\pm i)^{-1}$ and $hD_s(\mathfrak{D}_{h,0}\pm i)^{-1}$ are bounded from $L^2(\Omega_0,\mathbb{C}^2)$ to $L^2(\Omega_0,\mathbb{C}^2)$. Their adjoints $(\mathfrak{D}_{h}\pm i)^{-1}hD_s$ and $(\mathfrak{D}_{h,0}\pm i)^{-1}hD_s$ can be extended to become bounded operators from $L^2(\Omega_0,\mathbb{C}^2)$ to $L^2(\Omega_0,\mathbb{C}^2)$.
	Moreover, by Rellich–Kondrachov theorem, $(\mathfrak{D}_{h}\pm i)^{-1}$ and $(\mathfrak{D}_{h,0}\pm i)^{-1}$ are compact from $L^2(\Omega_0,\mathbb{C}^2)$ to $L^2_{\rm loc}(\Omega_0,\mathbb{C}^2)$. Since, $V_h$ and $W_h$ are bounded and compactly supported, we get that $W_h(\mathfrak{D}_{h}\pm i)^{-1}$, $V_h(\mathfrak{D}_{h}\pm i)^{-1}$, $W_h(\mathfrak{D}_{h,0}\pm i)^{-1}$, $V_h(\mathfrak{D}_{h,0}\pm i)^{-1}$ as well as their adjoints are compact operators from $L^2(\Omega_0,\mathbb{C}^2)$ to $L^2(\Omega_0,\mathbb{C}^2)$.
	Therefore, we obtain that
	\[\begin{split}
		(\mathfrak{D}_{h,0}+i)^{-1}hD_sW_h(\mathfrak{D}_{h}+i)^{-1} 
		&= \left[(\mathfrak{D}_{h,0}+i)^{-1}hD_s\right]\left[W_h(\mathfrak{D}_{h}+i)^{-1}\right]\,,\\
		(\mathfrak{D}_{h,0}+i)^{-1}W_hhD_s(\mathfrak{D}_{h}+i)^{-1} 
		&= \left[(\mathfrak{D}_{h,0}+i)^{-1}W_h\right]\left[hD_s(\mathfrak{D}_{h}+i)^{-1}\right] \,,\\
		(\mathfrak{D}_{h,0}+i)^{-1}V_h(\mathfrak{D}_{h}+i)^{-1}
		&=\left[(\mathfrak{D}_{h,0}+i)^{-1}\right]\left[V_h(\mathfrak{D}_{h}+i)^{-1}\right]\,,
	\end{split}\]
	are compact from $L^2(\Omega_0,\mathbb{C}^2)$ to $L^2(\Omega_0,\mathbb{C}^2)$.
	Therefore, $(\mathfrak{D}_h+i)^{-1}-(\mathfrak{D}_{h,0}+i)^{-1}$ is compact and the Weyl criterion gives the equality of the essential spectra.
\end{proof}

\subsection{A fibered family of Dirac operators}\label{sec.22}
By using the semiclassical Fourier transform, we see that
\begin{equation}\label{eq.direct}
\mathscr{D}_{h,0}=\int^\oplus\mathscr{D}_{h,0,\xi}\dd\xi\,,
\end{equation}
with
\[\mathscr{D}_{h,0,\xi}=(\xi+t)\sigma_1	+\sigma_2 D_t=\begin{pmatrix}
0&\xi-h\partial_{t}+t\\
\xi+h\partial_{t}+t&0
\end{pmatrix}\,,\]
with domain 
\[
	\mathrm{Dom}(\mathscr{D}_{h,0,\xi})=\{\psi = (\psi_1,\psi_2)\in H^1(I, \mathbb{C}^2),\psi_1(\pm\delta)=\mp\psi_2(\pm\delta) \}\,,
\]
where $I = (-\delta,\delta)$.

\begin{figure}[htb]
    \centering
    \includegraphics[width=1\textwidth]{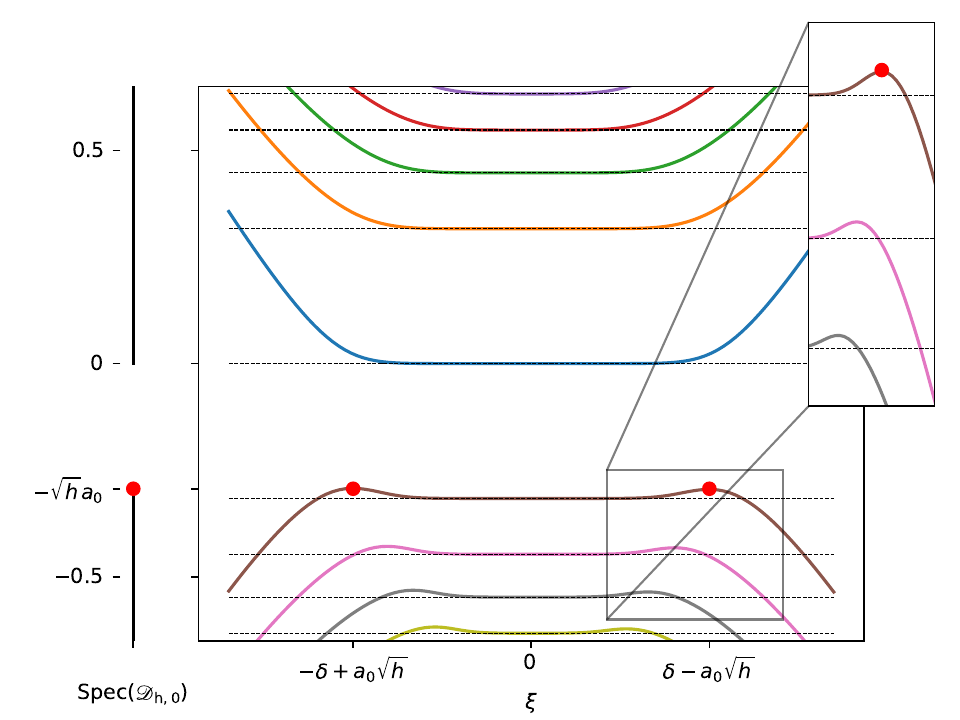}
    \caption{Dispersion curves for $h = 0.05$}
    \label{fig:dispCurves}
\end{figure}
In this section, we outline the standard spectral properties of the operators $(\mathscr{D}_{h,0,\xi})_\xi$. We provide only a brief proof of their invertibility and the symmetries of the associated dispersion curves, the latter being a consequence of the symmetries of $\Omega_0$. The proofs of the other stated properties are omitted, as they are very close to the existing literature (see for instance  \cite[Proposition 4.2]{BLTRS23}). 

The associated dispersion curves are illustrated in Figure \ref{fig:dispCurves}.
\begin{proposition}\label{prop.propDispCurves}
Let $h>0$. The following holds.
\begin{enumerate}[\rm (i)]
	\item\label{pt.prop.propDispCurves1} For $\xi\in\mathbb{R}$, the operator $\mathscr{D}_{h,0,\xi}$ is neither bounded from below nor from above, it is self-adjoint, inversible and has compact resolvent. Its eigenvalues are simple and denoted by
	\[\ldots\leq -\mu_2^-(\xi,h)\leq-\mu_1^-(\xi,h)<0<\mu_1^+(\xi,h)\leq \mu_2^+(\xi,h)\leq \ldots\,.\]
	\item \label{pt.prop.propDispCurves2}  For  $k\in\mathbb{N}\setminus\{0\}$, the map $\xi\mapsto \mu_k^{\pm}(\xi,h)$ is analytic and even.
\end{enumerate}
\end{proposition}

\subsubsection{Proof of the symmetry of the dispersion curves}

To prove that the dispersion curves are even, note that $\Omega_0$ is left stable by the point symmetry around $0$. The covariance of the Dirac operator is expressed on the fibered operators in the following lemma ($i\sigma_3$ being the rotor associated with the symmetry).
\begin{lemma}\label{lem.even}
Considering the unitary transformation $S : \psi\mapsto i\sigma_3\psi(-\cdot)$, we have, for all $\xi\in\R$, $S\mathrm{Dom}(\mathscr{D}_{h,0,\xi}) = \mathrm{Dom}(\mathscr{D}_{h,0,\xi})$ and
\[S^{*}\mathscr{D}_{h,0,\xi}S=\mathscr{D}_{h,0,-\xi}\,.
\]	
In particular, we have 
\[\mu^\pm_k(\xi,h)=\mu^\pm_k(-\xi,h)\,.
\]
\end{lemma}

\begin{remark}\label{rem.charge}
Note that the \emph{charge conjugation} for the bidimentionnal Dirac operator is $C\colon \psi \mapsto \sigma_1\overline{\psi}$. For the fibered ones, it reduces to a multiplication by $\sigma_1$ and to a change of sign for $\xi$. This leaves the domain $\mathrm{Dom}(\mathscr{D}_{h,0,\xi})$ stable. The operator $\mathscr{D}_{h,0,\xi}$ is transformed 
 into $-\left((\xi-t)\sigma_1+\sigma_2 D_t\right)$ (the magnetic field has opposite sign). 

\end{remark}

\subsubsection{Proof of the invertibility of the operators}
Since the spectrum is discrete, it is sufficient to consider the equation $\mathscr{D}_{h,0,\xi}\psi=0$. Let $\psi \in \ker(\mathscr D_{h,0,\xi})$. We have
\[(\xi-h\partial_{t}+t)\psi_2=0\,,\quad (\xi+h\partial_{t}+t)\psi_1=0\,.\] 	
Then,
\[\langle(\xi-h\partial_{t}+t)\psi_2,\psi_1\rangle=0\,,\]
and, by integration by parts,
\[\langle(\xi-h\partial_{t}+t)\psi_2,\psi_1\rangle=-h(\psi_2(\delta)\overline{\psi_1}(\delta)-\psi_2(-\delta)\overline{\psi_1}(-\delta))=0\,,\]
so that, using the boundary condition,  $\psi_1(\pm\delta)=\psi_2(\pm\delta)=0$. Thus, $\psi_{1}, \psi_2$ vanish at the boundary and solve first order linear ordinary differential equations, therefore $\psi = 0$.

\subsection{A characterization of the eigenvalues}
The next sections will focus on a more detailed and quantitative analysis of the dispersion curves. The key tool for this analysis is the nonlinear min-max characterization developed in \cite{BLTRS23}, adapted to our specific context (see \cite{BLTRS23} for references on min-max characterizations for eigenvalues of Dirac operator in domains without boundaries). This theory was primarily established in \cite[Section 2]{BLTRS23} for two-dimensional magnetic Dirac operators on bounded domains. In the simpler setting of the one-dimensional magnetic Dirac operator on the half-line \cite[Section 4]{BLTRS23}, similar results were obtained using the same arguments. These arguments can also be applied to Dirac operators on a bounded interval. Note also that Section \ref{sec.33} is dedicated to similar characterization arguments in the two-dimensional case of the strip, where the presence of essential spectrum complicates the proof.

The results for this setting are summarized in the following proposition.
To simplify the statement of the proposition, let us introduce the following notations.
\begin{notation}
Let $h>0$, $\lambda\geq 0$ and $\xi\in\R$.
	\begin{enumerate}
		\item
			For $\psi\in H^1(I)$, we define
	\[\mathscr{Q}^\pm_{\lambda,\xi,h}(\psi):=\|(\pm h\partial_t+\xi+t)\psi\|^2+\lambda h\|\psi\|^2_{\partial I}-\lambda^2\|\psi\|^2\,,\]
and $\ell_1^\pm(\lambda,\xi,h) = \inf_{\underset{\|\psi\|=1}{\psi\in H^1(I)}}\mathscr{Q}^\pm_{\lambda,\xi,h}(\psi)$. 
		\item For $\psi\in H^1(I)\setminus\{0\}$, we define
		\[
		\rho^\pm_{\xi,h}(\psi ) := \frac{h\|\psi\|^2_{\partial I}+\sqrt{h^2\|\psi\|^4_{\partial I}+4\|\psi\|^2\|(\pm h\partial_t+t+\xi)\psi\|^2}}{2\|\psi\|^2}\,.
		\]
	\end{enumerate}
\end{notation}
\begin{remark}
Let $h>0$, $\lambda\geq 0$ and $\xi\in\R$.
\begin{enumerate}
\item
 The operator associated to $\mathscr{Q}^\pm_{\lambda,\xi,h}$ is of compact resolvent and $\ell_1^\pm(\lambda,\xi,h)$ represents the corresponding lowest eigenvalue. 
 \item For $\psi\in H^1(I)\setminus\{0\}$, the mapping $\lambda\mapsto \mathscr{Q}^\pm_{\lambda,\xi,h}(\psi)$ is a quadratic polynomial with two roots : one that is non-positive and another, denoted by $\rho^\pm_{\xi,h}(\psi )$ which is positive.
 \end{enumerate}
\end{remark}
\begin{proposition}\label{prop.charact}
The following two characterizations of $\mu^\pm_1(\xi, h)$ hold.
\begin{enumerate}[\rm (i)]
\item We have
\begin{equation*}\label{eq.muroot}
\mu^\pm_1(\xi,h)=\min_{\psi\in H^1(I)\setminus\{0\}}\rho^\pm_{\xi,h}(\psi )\,.
\end{equation*}	
\item
$\mu^\pm_1(\xi,h)$ is the unique positive solution $\lambda$ of
\[\ell^\pm_1(\lambda,\xi,h)=0\,.\]
Moreover, we have
$\ell^\pm_1(\lambda,\xi,h)>0$ for all $\lambda\in(0,\mu^\pm_1(\xi,h))$ and $\ell^\pm_1(\lambda,\xi,h)<0$ for all $\lambda\in(\mu^\pm_1(\xi,h),+\infty)$.
\end{enumerate}
In addition, we have for $\lambda>0$,
\begin{equation}\label{eq.cont23}
|\ell^\pm_1(\lambda,\xi,h)|\geq \lambda|\mu_1^\pm(\xi,h)-\lambda|\,.	
\end{equation}
\end{proposition}
\begin{remark}\label{rem.proccomm}
The previous proposition guarantees that the following two procedures yield the same constant, which corresponds to the Dirac eigenvalue $\mu^\pm_1(\xi,h)$.
\begin{enumerate}[\rm (a)]
\item First, minimize $\mathscr{Q}^\pm_{\lambda,\xi,h}$ over normalized $H^1$ function to get $\ell^\pm_1(\lambda,\xi,h)$, and then find the positive root $\lambda>0$ of $\ell^\pm_1(\lambda,\xi,h)$.
\item First, find the root $\rho^\pm_{\xi,h}(\psi )$ of $\mathscr{Q}^\pm_{\lambda,\xi,h}(\psi)$ for any normalized $\psi\in H^1(I)$, and then, find the infimum $\lambda$ of the roots.
\end{enumerate}
\end{remark}
\begin{proof}
Firstly, note that since we are considering a compact set, all infima are attained as minima. 
\begin{enumerate}[\rm a)]
\item 
Let us begin by proving that the two procedures of Remark \ref{rem.proccomm} yield the same constant.
We have that $\ell_1^\pm(0,\xi,h) = 0$, the function $\lambda\mapsto \ell_1^\pm(\lambda,\xi,h)$ is concave, and
\[
\limsup_{\lambda\to+\infty}\ell_1^\pm(\lambda,\xi,h)\leq \limsup_{\lambda\to+\infty} \left(-\lambda^2 + \inf_{\underset{\|\psi\|=1}{\psi\in H^1_0(I)}}\|(\pm h\partial_t+\xi+t)\psi\|^2\right) = -\infty.
\]
Moreover, the constant 
\begin{equation*}
\nu^\pm_1(\xi,h):=\min_{\psi\in H^1(I)\setminus\{0\}}\rho^\pm_{\xi,h}(\psi )\,,
\end{equation*}	
is positive.
By definition, for any normalized $\psi\in H^1$, we have $\nu^\pm_1(\xi,h)\leq \rho^\pm_{\xi,h}(\psi )$ so that the quadratic polynomial $\lambda\mapsto \mathscr{Q}^\pm_{\lambda,\xi,h}(\psi)$ is positive on $(0, \nu^\pm_1(\xi,h))\subset (0,\rho^\pm_{\xi,h}(\psi ))$. Therefore, for all $\lambda \in (0, \nu^\pm_1(\xi,h))$, we have
\[
	\ell_1^\pm(\lambda,\xi,h) = \inf_{\underset{\|\psi\|=1}{\psi\in H^1(I)}}\mathscr{Q}^\pm_{\lambda ,\xi,h}(\psi)> 0.
\] 
The inequality is strict because infimum is attained as a minimum. 
Now, let $\psi\in H^1$ be a normalized minimizer of $\nu^\pm_1(\xi,h)$. The quadratic polynomial $\lambda\mapsto \mathscr{Q}^\pm_{\lambda,\xi,h}(\psi)$ is negative on $(\nu^\pm_1(\xi,h),+\infty)$, so that for all $\lambda>\nu^\pm_1(\xi,h)$, 
\[
	\ell_1^\pm(\lambda,\xi,h)\leq \mathscr{Q}^\pm_{\lambda ,\xi,h}(\psi)<0\,.
\]
 This proves that $\nu^\pm_1(\xi,h)$ is the unique positive root to the equation $\ell_1^\pm(\lambda,\xi,h) = 0$.
\item Let $0<\lambda_1<\lambda_2$ and $\psi\in H^1$ be a normalized function. We have
\[\begin{split}
&\left(\lambda_1^{-1} \mathscr{Q}^\pm_{\lambda_1 ,\xi,h}(\psi) - \lambda_2^{-1} \mathscr{Q}^\pm_{\lambda_2 ,\xi,h}(\psi)\right)
\\&= \frac{\lambda_2-\lambda_1}{\lambda_1\lambda_2}\|(\pm h\partial_t+\xi+t)\psi\|^2
+(\lambda_2-\lambda_1)\|\psi\|^2
\\&\geq (\lambda_2-\lambda_1)\|\psi\|^2\,,
\end{split}\]
so that $\lambda^{-1}_1\ell_1^\pm(\lambda_1,\xi,h)-\lambda^{-1}_2\ell_1^\pm(\lambda_2,\xi,h)\geq \lambda_2-\lambda_1$. Let $\lambda>0$. Taking $(\lambda_1,\lambda_2) = (\lambda, \nu^\pm_1(\xi,h))$  and $(\lambda_1,\lambda_2) = (\nu^\pm_1(\xi,h), \lambda)$ in the previous inequality implies that
\begin{equation*}
|\ell^\pm_1(\lambda,\xi,h)|\geq \lambda|\nu_1^\pm(\xi,h)-\lambda|\,.	
\end{equation*}
\item Let us prove now that $\nu^+_1(\xi,h)$ is also the lowest positive eigenvalue of the Dirac operator. The same arguments apply for the first negative eigenvalue. 
Let $\varphi = (\varphi_1,\varphi_2)\in \mathrm{Dom}(\mathscr{D}_{h,0,\xi})$ be an eigenfunction of the operator $\mathscr{D}_{h,0,\xi}$ associated with the smallest positive eigenvalue $\lambda := \mu_1^+(\xi,h)$. Then, we have
\begin{equation}\label{eq.eig}
\begin{array}{l}
	(\xi+t +h\partial_t)\varphi_1 = \lambda \varphi_2\,,
	\\(\xi+t -h\partial_t)\varphi_2 = \lambda \varphi_1\,.
\end{array}
\end{equation}
Additionally, by integrating by parts, we obtain
\[\begin{split}
	\|(\xi+t +h\partial_t)\varphi_1\|^2 
	&= \braket{(\xi+t +h\partial_t)\varphi_1,\lambda\varphi_2}_{L^2(I)} 
	\\
	&= \lambda\braket{\varphi_1,(\xi+t -h\partial_t)\varphi_2}_{L^2(I)} + h\lambda[\braket{\varphi_1,\varphi_2}_{\C}]^\delta_{-\delta}
	\\
	&= \lambda^2\|\varphi_1\|^2 - h\lambda\|\varphi_1\|^2_{\partial I}\,.
\end{split}\]
For the last equality, the boundary condition was used. This shows that 
\[
\ell_1^+(\lambda,\xi,h)\leq \mathscr{Q}^+_{\lambda,\xi,h}(\varphi_1) = 0\,,
\]
so that $\nu^+_1(\xi,h)\leq \lambda =\mu_1^+(\xi,h)$. 
The conclusion follows once we prove that $\nu^+_1(\xi,h)$ is an eigenvalue of the Dirac operator. To proceed, let $\overline{\lambda} = \nu^+_1(\xi,h)$ and let $\varphi_1\in H^1(I)$ be a normalized minimizer of $\ell_1^+(\overline{\lambda},\xi,h) = 0$. The function $\varphi_1$ satisfies
\[\begin{array}{l}
	\left(
		-h^2\partial_t + (\xi+t)^2 - h - \overline{\lambda}^2
	\right) \varphi_1 = 0 \text{ on }I\,,
	\\
	\overline{\lambda}^{-1}\left(
		h\partial_t+\xi+t
	\right)\varphi_1(\pm \delta) = \mp\varphi_1(\pm \delta)\text{ on }\partial I\,.
\end{array}\]
Let us remark that $\varphi_1\in H^2(I)$ and that $\varphi := (\varphi_1, \overline{\lambda}^{-1}(h\partial_t+\xi+t)\varphi_1)$ is an eigenfunction of the Dirac operator associated with the eigenvalue $\overline{\lambda}$. This concludes the proof.

\end{enumerate}
\end{proof}

\subsection{\texorpdfstring{Study of $\mu^+_1(\xi,h)$}{Study of the first positive dispersion curve}}

The following proposition presents some properties of the first positive dispersion curve.
\begin{proposition}\label{prop.curves+}
		\begin{enumerate}[\rm (a)]
			\item \label{pt.propdisp1}For all $\xi$, we have $\mu^+_1(\xi,h) \leq \nu_1(\xi,h)$\,,
			with 
			\[
				\nu_1(\xi,h):=\frac{h(e^{-\frac{(\xi-\delta)^2}{h}}+e^{-\frac{(\xi+\delta)^2}{h}})}{\int_{-\delta}^\delta e^{-(\xi+t)^2/h}\dd t}\,.
			\]
			\item \label{pt.propdisp2}
			$
				\inf_{\xi\in\mathbb{R}}\nu_1(\xi,h) = \nu_1(0,h)=2\sqrt{\frac{h}{\pi}}e^{-\frac{\delta^2}{h}}(1+o_{h\to0}(1))\,.
			$
			\item \label{pt.propdisp3} For all $|\xi|>\delta$, 
\[\mu_1^+(\xi,h)\geq |\xi|-\delta-\frac{h}{|\xi|-\delta}\,.\]
			\item \label{pt.propdisp4} For all $\xi$ such that $\mu_1^+(\xi,h) \leq  2\sqrt{h}e^{-\delta^2/h}$ we have
			\[
				\mu_1^+(\xi,h)\geq \nu_1(\xi,h)\left(1+o_{h\to0}(1)\right)
			\]
			with $o_{h\to0}(1)$ uniform in $\xi$.
		\end{enumerate}
\end{proposition}

Proposition \ref{prop.curves+} gathers the ingredient to characterize the positive part of the spectrum of $\mathscr{D}_{h,0}$ as stated in Theorem \ref{thm.ess}. 
\begin{corollary}\label{cor.lambdaess+s}
We have
\[
	\mathrm{sp}(\mathscr{D}_{h,0})\cap [0,+\infty[ = [\lambda^+_{\rm ess}(h),+\infty[\,,
\]
with $\lambda^+_{\rm ess}(h) :=\inf_{\xi\in\mathbb{R}}\mu_1^+(\xi,h) =  2\sqrt{\frac{h}{\pi}}e^{-\delta^2/h}(1+o_{h\to0}(1))$.
\end{corollary}
\begin{proof}[Proof of Corollary \ref{cor.lambdaess+s}]
	Let $h>0$.
	By \cite[Theorems XIII.85]{MR493421},
	we have 
	\[
		\{\mu_1^+(\xi,h)\,,\xi\in\mathbb{R}\} \subset \mathrm{sp}(\mathscr{D}_{h,0})\cap [0,+\infty[ \,,
	\]
	and
	\begin{equation}\label{eq.infspec12}
		\inf_{\xi\in\mathbb{R}}\mu_1^+(\xi,h) = \inf\mathrm{sp}(\mathscr{D}_{h,0})\cap [0,+\infty[\,.
	\end{equation}
	By Point \eqref{pt.propdisp3} of Proposition \ref{prop.curves+}, $\xi\mapsto \mu_1^+(\xi,h)$ is coercive : 
	\[
		\lim_{|\xi|\to+\infty}\mu_1^+(\xi,h) = +\infty.
	\]
	Therefore, there exists $\xi_h$ such that $\inf_{\xi\in\R} \mu_1^+(\xi_h,h) = \mu_1^+(\xi_h,h)$ and
	\begin{equation*}
		\{\mu_1^+(\xi,h)\,,\xi\in\mathbb{R}\}  = [\inf_{\xi\in\R}\mu_1^+(\xi,h),+\infty[\subset\mathrm{sp}(\mathscr{D}_{h,0})\cap [0,+\infty[  \,.
	\end{equation*}
	 This, together with \eqref{eq.infspec12}, implies that
	\begin{equation*}\label{eq.specess111}
		[\inf_{\xi\in\R}\mu_1^+(\xi,h),+\infty[ = \mathrm{sp}(\mathscr{D}_{h,0})\cap [0,+\infty[ \,.
	\end{equation*}
	By Points \eqref{pt.propdisp1} and \eqref{pt.propdisp2} of Proposition \ref{prop.curves+}, we also have that
	\begin{equation}\label{eq.estupp}
	\inf_{\xi\in\R}\mu_1^+(\xi,h)\leq \inf_{\xi\in\R}\nu_1(\xi,h)= \nu_1(0,h)\leq 2\sqrt{\frac{h}{\pi}}e^{-\delta^2/h}(1+o_{h\to0}(1))\,.
	\end{equation}
	To get the lower bound, note that \eqref{eq.estupp} and Point \eqref{pt.propdisp4} of Proposition \ref{prop.curves+} ensure
	\[\begin{split}
	\inf_{\xi\in\R}\mu_1^+(\xi,h)
	&= \mu_1^+(\xi_h,h)
	\geq \nu_1(\xi_h,h)(1+o_{h\to0}(1))
	\\&\geq (1+o_{h\to0}(1))\inf_{\xi\in\R}\nu_1(\xi,h) = (1+o_{h\to0}(1))\nu_1(0,h)\,.
	\end{split}\]
	The conclusion follows from Point \eqref{pt.propdisp2} of Proposition \ref{prop.curves+}.
\end{proof}

\subsubsection{Proof of Point \eqref{pt.propdisp1} and \eqref{pt.propdisp2} of Proposition \ref{prop.curves+}}\label{sec.23}
By considering the function $u_{\xi}(t)=e^{-\frac{(\xi+t)^2}{2h}}$, Proposition \ref{prop.charact} implies Point \eqref{pt.propdisp1}.

To prove \eqref{pt.propdisp2}, remark that 
\[\nu_1(\xi,h)=he^{-\delta^2/h}\frac{e^{2\xi\delta/h}+e^{-2\xi\delta/h}}{\int_{-\delta}^\delta e^{-2\xi t/h}e^{-t^2/h}\dd t}\,,\]	
and 
\begin{multline*}
\partial_\xi\nu_1(\xi,h)=\frac{2e^{-\delta^2/h}}{\left(\int_{-\delta}^\delta e^{-2\xi t/h}e^{-t^2/h}\dd t\right)^2}\Big[(e^{2\xi\delta/h}-e^{-2\xi\delta/h})\int_{-\delta}^\delta \delta e^{-2\xi t/h}e^{-t^2/h}\dd t\\
+(e^{2\xi\delta/h}+e^{-2\xi\delta/h})\int_{-\delta}^\delta t e^{-2\xi t/h}e^{-t^2/h}\dd t \Big]\,.
\end{multline*}
The quantity between the brackets equals
\[\int_{-\delta}^\delta(t+\delta)e^{(\delta-t)\xi/h}e^{-t^2/h}\mathrm{d}t+\int_{-\delta}^\delta(t-\delta)e^{-(\delta+t)\xi/h}e^{-t^2/h}\mathrm{d}t\,,\]
and also (by using $t\mapsto -t$ in the second integral)
\begin{multline*}
\int_{-\delta}^\delta(t+\delta)e^{(\delta-t)\xi/h}e^{-t^2/h}\mathrm{d}t-\int_{-\delta}^\delta(t+\delta)e^{-(\delta-t)\xi/h}e^{-t^2/h}\mathrm{d}t\\
=2\int_{-\delta}^\delta(t+\delta)\sinh((\delta-t)\xi/h)e^{-t^2/h}\mathrm{d}t\,.
\end{multline*}
Thus,
\begin{equation*}
\partial_\xi\nu_1(\xi,h)=\frac{4e^{-\delta^2/h}}{\left(\int_{-\delta}^\delta e^{-2\xi t/h}e^{-t^2/h}\dd t\right)^2}\int_{-\delta}^\delta(t+\delta)\sinh((\delta-t)\xi/h)e^{-t^2/h}\mathrm{d}t\,,
\end{equation*}
which is positive when $\xi>0$. Therefore, the function $\xi\mapsto\nu_1(\xi,h)$ is even, increasing on $\R_+$ and Point \eqref{pt.propdisp2} follows.

\subsubsection{Proof of Point \eqref{pt.propdisp3} of Proposition \ref{prop.curves+}}\label{sec:proofc}

We consider the case $\xi > \delta$.
Let us use Proposition \ref{prop.charact} and consider $\ell_1^+(\lambda,\xi,h)$ with $\lambda=\xi-\delta>0$. We have
\[
\mathscr{Q}^+_{\lambda,\xi,h}(\psi)=\|h\partial_t\psi\|^2+\|(\xi+t)\psi\|^2+2h\Re\langle \partial_t\psi,(\xi+t)\psi\rangle+\lambda h\|\psi\|^2_{\partial I}-\lambda^2\|\psi\|^2\,,
\]
and, by integration by parts,
\begin{multline*}
\mathscr{Q}^+_{\lambda,\xi,h}(\psi)=\|h\partial_t\psi\|^2+\|(\xi+t)\psi\|^2-h\|\psi\|^2\\
+h(\xi+\delta)|\psi(\delta)|^2-h(\xi-\delta)|\psi(-\delta)|^2+\lambda h\|\psi\|^2_{\partial I}-\lambda^2\|\psi\|^2\,.
\end{multline*}
Since $\xi > \delta$ we have $\|(\xi+t) \psi\|^2 \geq \lambda^2 \|\psi\|^2$. We deduce that
\begin{equation*}
	\mathscr{Q}^+_{\lambda,\xi,h}(\psi)\geq -h\|\psi\|^2\,,
\end{equation*}
so that
\[\ell_1^+(\xi-\delta,\xi,h)\geq -h\,.\]
If $\ell_1^+(\xi-\delta,\xi,h)\geq 0$, Proposition \ref{prop.charact} ensures that $\tilde \lambda\mapsto \ell_1^+(\tilde \lambda,\xi,h)$ is positive on $(0,\xi-\delta)$ and that its unique positive root $\mu_1^+(\xi,h)$ satisfies $\mu_1^+(\xi,h)\geq \xi-\delta$.

If $\ell_1^+(\xi-\delta,\xi,h)<0$, we get $|\ell_1^+(\xi-\delta,\xi,h)|\leq h$ and by \eqref{eq.cont23} (Proposition \ref{prop.charact}),
\[(\xi-\delta)|\mu^+_1(\xi,h)-(\xi-\delta)|\leq h\,.\]
We observe that the symmetry $(\xi,t) \mapsto (-\xi,-t)$ ensures that $\ell^+_1(\lambda, \xi, h) = \ell^+_1(\lambda, -\xi, h)$ for all $\xi$, $h$, and $\lambda$. Consequently, the case $\xi < -\delta$ follows.

\subsubsection{Proof of Point \eqref{pt.propdisp4} of Proposition \ref{prop.curves+}}
Let us turn to lower bounds for $\mu_1^+(\cdot,h)$. By Point \eqref{pt.propdisp1} and \eqref{pt.propdisp2}, there is $h_0>0$ such that for $h\in(0,h_0)$,
\begin{equation}\label{eq.setmin}
	\Xi_h:=\left\{\xi\geq 0 : \mu_1^+(\xi,h) \leq  2\sqrt h e^{-\frac {\delta^2}{h}}\right\}\ne\emptyset\,.
\end{equation}

\begin{notation}
We consider on $L^2(I)$ the operator 
\[
d_{h,\xi} = - h \partial_t + \xi + t 
\]
with domain  
\[
\mathrm{Dom}(d_{h,\xi}) = H^1_0(I)\,.
\]
Its adjoint is 
\[
d^*_{h,\xi} = h \partial_t + \xi + t
\]
with domain 
\[
\mathrm{Dom}(d_{h,\xi}^*) = H^1(I).
\]
We denote by $\Pi_\xi$ the orthogonal projection on $\mathrm{Ker}(d_{h,\xi}^*)$, which is spanned by $t\mapsto e^{-(\xi+t)^2/2h}$.
\end{notation}
\begin{notation}
Let $\xi\in\R$. By Proposition \ref{prop.charact}, we consider $\psi_\xi \in H^1(I)$ such that $\|\psi_\xi\| = 1$ and
\[
\mu_1^+(\xi,h) = \frac  h 2 \|\psi_\xi\|^2_{\partial I}+ \frac 12 \sqrt{h^2\|\psi_\xi\|^4_{\partial I}+4 \|(h\partial_t+t+\xi)\psi_\xi\|^2}\,.
\]	
\end{notation}

\begin{lemma}\label{lem.approxpsixi}
There exist $h_0>0$ and $C>0$ such that, for all $h\in(0,h_0)$ and $\xi\in\Xi_h$, 
\[\|\psi_\xi-\Pi_\xi\psi_\xi\|_{H^1(I)}\leq Ch^{-\frac32}e^{-\delta^2/h}\,.\]
\end{lemma}
\begin{proof}
For all $\xi\in\Xi_h$, we have
	\[
	\| d_{h,\xi}^* (\psi_\xi - \Pi_\xi \psi_\xi) \| = \| d_{h,\xi}^* \psi_\xi \|\leq \mu_1^+(\xi,h) \leq 2\sqrt h e^{-\frac {\delta^2}{h}}\,.
	\]
	Since $d_{h,\xi}$ has closed range, we have $\psi_\xi - \Pi_\xi \psi_\xi \in \mathrm{Ker}(d_{h,\xi}^*)^\bot = \mathrm{Ran} (d_{h,\xi})$. Let $\varphi_\xi \in H^2(I) \cap H_0^1(I)$ such that $d_{h,\xi} \varphi_\xi = \psi_\xi - \Pi_\xi \psi_\xi$. We consider the operator
	\[
	\mathscr L_{h,\xi}^+ = d_{h,\xi}^* d_{h,\xi} = -h^2 \partial_t^2 + (\xi+t)^2 + h \,,
	\]
	with domain $H^2(I) \cap H_0^1(I)$. We have $\| \mathscr L _{h,\xi}^+ \varphi_\xi \| \leq 2h^{\frac 12} e^{-\frac {\delta^2}{h}}$. 

Note that, 	since $\varphi_\xi$ satisfies the Dirichlet boundary condition, we have 
	\begin{multline*}
	\| \mathscr L_{h,\xi}^+ \varphi_\xi \|^2 
	 = \| h^2 \partial_t^2 \varphi_\xi \|^2 + \|(\xi+t)^2 \varphi_\xi \|^2 + \| h \varphi_\xi \|^2 + 2h^3 \| \partial_t \varphi_\xi \|^2  + 2 h \| (\xi+t) \varphi_\xi \|^2
	 \\+ 2 \Re \left\langle -h^2 \partial_t^2 \varphi_\xi , (\xi+t)^2 \varphi_\xi \right\rangle \,.
	\end{multline*}
Moreover,
\[2\Re \left\langle -h^2 \partial_t^2 \varphi_\xi , (\xi+t)^2 \varphi_\xi \right\rangle
=2h^2 \Re \left\langle \partial_t \varphi_\xi , \partial_t\left[(\xi+t)^2 \varphi_\xi\right] \right\rangle\]
and
\begin{align*}
	2h^2 \Re \left\langle \partial_t \varphi_\xi , \partial_t[ (\xi+t)^2 \varphi_\xi] \right\rangle
	& = 2h^2 \| (\xi+t) \partial_t \varphi_\xi \|^2 + 4 h^2 \Re \left\langle \partial_t \varphi_\xi, (\xi+t) \varphi_\xi \right\rangle\\
	&= 2h^2 \| (\xi+t) \partial_t \varphi_\xi \|^2 -  2h^2 \| \varphi_\xi \|^2\,.
	\end{align*}
	Therefore,
	\begin{multline*}
	\| \mathscr L_{h,\xi}^+ \varphi_\xi \|^2 
	 = \| h^2 \partial_t^2 \varphi_\xi \|^2 
	 + \|(\xi+t)^2 \varphi_\xi \|^2 
	 +h^2\|\varphi_\xi\|^2
	 + 2h^3 \| \partial_t \varphi_\xi \|^2
	 + 2 h \| (\xi+t) \varphi_\xi \|^2
	 \\
	 +2h^2 \| (\xi+t) \partial_t \varphi_\xi \|^2
	   -2h^2\|\varphi_\xi\|^2\,.
	\end{multline*}
	Thus,
		\begin{multline}\label{eq.L+2}
		\| \mathscr L_{h,\xi}^+ \varphi_\xi \|^2 
		= \| h^2 \partial_t^2 \varphi_\xi \|^2 + \|(\xi+t)^2 \varphi_\xi \|^2 
		+h^2\|\varphi_\xi\|^2+2h^2 \| (\xi+t) \partial_t \varphi_\xi \|^2\\
		+ 2h \| (h\partial_t+\xi+t)\varphi_\xi\|^2\,.
	\end{multline}
	and in particular
	\[h^4\|\partial_t^2\varphi_\xi\|^2+h^2\|\varphi_\xi\|^2\leq (2h^{1/2} e^{-\frac {\delta^2}{h}})^2\,.\]
	Since $\varphi_\xi\in H^2(I)\cap H^1_0(I)$, we get
		\[\|\varphi_\xi\|^2_{H^2(I)}\leq Ch^{-3} e^{-\frac {2\delta^2}{h}}\,.\]
	Then, we have
	\[\begin{split}
	\|d_{h,\xi}\varphi_\xi\|_{H^1(I)}&=\|(-h\partial_t+\xi+t)\varphi_\xi\|_{H^1(I)}\\
	&\leq C\left(h\|\partial_t\varphi_\xi\|_{H^1} + \|(\xi+t)\varphi_\xi\|_{H^1}\right)
	\\
	&\leq C\left(\|\varphi_\xi\|_{H^2(I)}+\|(\xi+t)\varphi_\xi\|+\|(\xi+t)\partial_t\varphi_\xi\|\right)\\
	&\leq \tilde Ch^{-\frac32}e^{-\delta^2/h}\,,
	\end{split}\]
where we used \eqref{eq.L+2} and the Cauchy-Schwarz inequality.
\end{proof}
Now, we can prove Point \eqref{pt.propdisp4}. We have
\[
\mu_1^+(\xi,h)\geq h \|\psi_\xi\|^2_{\partial I}\,.
\]
By Lemma \ref{lem.approxpsixi} and a usual trace theorem, considering $h$ small enough and $\xi\in\Xi_h$, we get
\[\|\psi_\xi-\Pi_\xi\psi_\xi\|+\|\psi_\xi-\Pi_\xi\psi_\xi\|_{\partial I}
\leq Ch^{-\frac32}e^{-\delta^2/h}\,,\]
so
\[
\|\psi_\xi\|_{\partial I}^2 = \frac{\|\Pi_\xi\psi_\xi\|_{\partial I}^2}{\|\Pi_\xi\psi_\xi\|^2}\left(1+o_{h\to0}(1)\right).
\]
Finally,
\[ \frac{\|\Pi_\xi\psi_\xi\|_{\partial I}^2}{\|\Pi_\xi\psi_\xi\|^2} = \nu_1(\xi,h),
\]
and the conclusion follows.

\subsection{\texorpdfstring{Study of $\mu^-_1(\xi,h)$}{Study of the first negative dispersion curve}}\label{sec.25}

In this section, we continue the proof of Theorem \ref{thm.ess} by establishing the estimate of $\lambda^-_{\mathrm{ess}}(h)$.
\begin{proposition}\label{prop.curves-}
We have
		\begin{enumerate}[\rm (a)]
			\item \label{pt.propdisp-3} There exist \( C \) and \( h_0 > 0 \) such that, for all \( h \in (0, h_0] \) and all $|\xi|>\delta$, 
\[\mu_1^-(\xi,h)\geq |\xi|-\delta-\frac{C}{|\xi|-\delta}\,.\]
			\item \label{pt.propdisp-4} For any $N\in\N$, there exists $C>0$ such that for all $\xi$, and all $h>0$,
			\[
				\mu_1^-(\xi,h)\geq a_0\sqrt{h} - Ch^{N}\,.
			\]
			with $C$ independent of $\xi, h$.
			
			\item \label{pt.propdisp-5} We have
			\[
\mu^-_1( \delta - \sqrt{h}a_0,h)\leq  a_0\sqrt{h} + \mathscr{O}(h^\infty)\,.
\]
		\end{enumerate}

\end{proposition}

This proposition \ref{prop.curves-} allows us to study the negative part of the spectrum of $\mathscr{D}_{h,0}$ as stated in Theorem \ref{thm.ess}. 
\begin{corollary}
We have
\[
	\mathrm{sp}(\mathscr{D}_{h,0})\cap ]-\infty,0] = (-\infty, -\lambda^-_{\rm ess}(h)]\,,
\]
with $\lambda^-_{\rm ess}(h) :=\inf_{\xi\in\mathbb{R}}\mu_1^-(\xi,h) =  \sqrt{h}a_0 + \mathscr{O}(h^{\infty})$.
\end{corollary}

\subsubsection{Proof of Point \eqref{pt.propdisp-3} of Proposition \ref{prop.curves-}}

As in Section \ref{sec:proofc}, we consider the case $\xi > \delta$ and take $\lambda=\xi-\delta$. We have, for all $\psi\in H^1(I)$,
\begin{multline*}
\mathscr{Q}^-_{\lambda,\xi,h}(\psi)=\|h\partial_t\psi\|^2+\|(\xi+t)\psi\|^2+h\|\psi\|^2\\
-h(\xi+\delta)|\psi(\delta)|^2+h(\xi-\delta)|\psi(-\delta)|^2+\lambda h\|\psi\|^2_{\partial I}-\lambda^2\|\psi\|^2\,.
\end{multline*}
Since $\xi>\delta$, we have for $t\in(-\delta,\delta)$, $(\xi+t)^2\geq (\xi-\delta)^2 = \lambda^2$. Moreover,
\[
-h(\xi+\delta)|\psi(\delta)|^2+h(\xi-\delta)|\psi(-\delta)|^2+\lambda h\|\psi\|^2_{\partial I} = -2h\delta|\psi(\delta)|^2+2h(\xi-\delta)|\psi(-\delta)|^2\,,
\]
so that
\[
\mathscr{Q}^-_{\lambda,\xi,h}(\psi)\geq \|h\partial_t\psi\|^2+h\|\psi\|^2-2h\delta|\psi(\delta)|^2\,.
\]
By Sobolev embedding, we recall that there exists $C>0$ such that, for all $\epsilon>0$,
	\[\|\psi\|^2_{L^\infty(I)}\leq \epsilon\|\partial_t\psi\|^2+C\epsilon^{-1}\|\psi\|^2\,.\]
	Thus, with $\epsilon=h(2\delta)^{-1}$, we deduce that
	\[\mathscr Q^-_{\lambda, \xi, h} (\psi)\geq-4C\delta^2\|\psi\|^2\,.\]
Therefore, by the min-max theorem,
\[\ell_1^{-}(\lambda, \xi, h)\geq -4C\delta^2\,.\]
By Proposition \ref{prop.charact}, if $\ell_1^{-}(\lambda, \xi, h)\geq 0$, then $\tilde \lambda\mapsto \ell_1^{-}(\tilde \lambda, \xi, h)$ is positive on $(0,\lambda)$ and $\mu_1^-(\xi,h)\geq \lambda=\xi-\delta$. If  $\ell_1^{-}(\lambda, \xi, h)\leq 0$, then $|\ell_1^{-}(\lambda, \xi, h)|\leq 4C\delta^2$. Thus, by \eqref{eq.cont23} (Proposition \ref{prop.charact}),
\[\mu_1^-( \xi, h)\geq \xi-\delta-4C\delta^2(\xi-\delta)^{-1}\,.\]

\subsubsection{Proof of Point \eqref{pt.propdisp-4} of Proposition \ref{prop.curves-}}

Let us denote for $u\in \{v\in H^1(\mathbb{R})\,,tv(t)\in L^2(\mathbb{R})\}$,
\[
\begin{split}
	\mathscr{Q}^-_{\lambda,\xi,\mathbb{R}}(u):= \|(-\partial_t+\xi+t)u\|^2_{L^2(\mathbb{R})} - \lambda^2\|u\|^2_{L^2(\mathbb{R})}\,,
\end{split}
\]
and $u\in \{v\in H^1(\mathbb{R}_+)\,,tv(t)\in L^2(\mathbb{R}_+)\}$,
\[
\begin{split}
	\mathscr{Q}^-_{\lambda,\xi,\mathbb{R}_+}(u):= \|(-\partial_t+\xi+t)u\|^2_{L^2(\mathbb{R}_+)} + \lambda|u(0)|^2 - \lambda^2\|u\|^2_{L^2(\mathbb{R}_+)}\,.
\end{split}
\]

\begin{remark}
As in Proposition \ref{prop.charact}, the quadratic forms $\mathscr{Q}^-_{\lambda,\xi,\mathbb{R}}$ and $\mathscr{Q}^-_{\lambda,\xi,\mathbb{R}_+}$ arise in the analysis of the negative spectrum of the magnetic Dirac operator on the domains $\mathbb{R}^2$ and $\mathbb{R} \times \mathbb{R}+$, with an infinite mass boundary condition and a magnetic field of strength 1 (after fibration). The notation in \cite[Notation 10]{BLTRS23} compares with ours as follows: for $u$ in the form domain,
\[\begin{split}
	&\mathscr{Q}^-_{\lambda,\xi,\mathbb{R}_+}(u) = q^-_{\R_+,\lambda,-\xi}(u)  - \lambda^2\|u\|^2
	\\&
	\mathscr{Q}^-_{\lambda,\xi,\mathbb{R}}(u) = q^-_{\R,-\xi}(u)  - \lambda^2\|u\|^2
	\,.
\end{split}\]
The bottom of the spectrum of the operator associated  to $\mathscr{Q}^-_{\lambda,\xi,\mathbb{R}}$ equals the Landau level $2-\lambda^2$ (\cite[Theorem 4.3]{BLTRS23}). In the study of  $\mathscr{Q}^-_{\lambda,\xi,\mathbb{R}_+}$, the constant $a_0$ plays a special role. By \cite[Propositions 4.12, 4.15]{BLTRS23}, the map $\xi\mapsto \inf_{v\ne0}\frac{1}{\|v\|^2}\mathscr{Q}^-_{a_0, \xi, \mathbb{R}_+}(v)$ is non-negative and has a unique non-degenerate minimum at $\xi = -a_0$, which is zero.
\end{remark}

Let $(\chi_1 ,\chi_2, \chi_3)$ be a partition of the unity of $[-\delta,\delta]$ such that $\mathrm{supp}(\chi_1) \subset [-\delta , -\delta/3]$,  $\mathrm{supp}(\chi_2) \subset [-2\delta/3 , 2\delta/3]$, $\mathrm{supp}(\chi_3) \subset [\delta/3, \delta]$,
\[
\chi_1^2 + \chi^2_2+ \chi_3^2 =1 \; \; \text{and} \; \; \vert \partial_t \chi_1 \vert^2 + \vert \partial_t \chi_2 \vert^2+ \vert \partial_t \chi_3 \vert^2  \leq C\,.
\]
From the localization formula, we have
\[\begin{split}
\| (-h\partial_t + \xi +t)\psi \|^2 
&
= \sum_{j=1}^{3} \left( \| (-h\partial_t + \xi +t)(\chi_j \psi) \|^2 - \| h(\partial_t \chi_j)\psi \|^2 \right)
\\
&\geq 
- Ch^2\|\psi \|^2
+\sum_{j=1}^{3} \| (-h\partial_t + \xi +t)(\chi_j \psi) \|^2 \,,
\end{split}\]
so that
\[
	\mathscr Q^-_{\lambda,\xi,h}(\psi)\geq - Ch^2\|\psi \|^2+\sum_{j=1}^{3}\mathscr Q^-_{\lambda,\xi,h}(\chi_i\psi).
\]
We introduce $v_1\colon s\in\mathbb{R}_+\mapsto (\chi_1\psi)(s\sqrt{h}-\delta)$, $v_2\colon s\in\mathbb{R}\mapsto (\chi_2\psi)(s\sqrt{h})$ and $v_3\colon s\in\mathbb{R}_+\mapsto (\chi_3\psi)(\delta - s\sqrt{h})$ so that
\[\begin{split}
&\mathscr Q^-_{\lambda,\xi,h}(\chi_1\psi) = h^{3/2}\mathscr Q^-_{\frac{\lambda}{\sqrt{h}}, \frac{\xi-\delta}{\sqrt{h}}, \mathbb{R}_+}(v_1)\,,
\\
&\mathscr Q^-_{\lambda,\xi,h}(\chi_2\psi) = h^{3/2}\mathscr Q^-_{\frac{\lambda}{\sqrt{h}}, \frac{\xi}{\sqrt{h}}, \mathbb{R}}(v_2)\,,
\\
&\mathscr Q^-_{\lambda,\xi,h}(\chi_3\psi) = h^{3/2}\mathscr Q^-_{\frac{\lambda}{\sqrt{h}}, -\frac{\xi+\delta}{\sqrt{h}}, \mathbb{R}_+}(v_3)\,.
\end{split}\]
Now, fix $\lambda = a_0\sqrt{h}$. By \cite[Propositions 4.12, 4.15]{BLTRS23}, the map $\xi\mapsto \inf_{v\ne0}\frac{1}{\|v\|^2}\mathscr{Q}^-_{a_0, \xi, \mathbb{R}_+}(v)$ is non-negative and has a unique non-degenerate minimum at $\xi = -a_0$, which is zero. By \cite[Theorem 4.3]{BLTRS23}, the map $\xi\mapsto \inf_{v\ne0}\frac{1}{\|v\|^2}\mathscr{Q}^-_{a_0, \xi, \mathbb{R}}(v)$ is constant equal to $2-a_0^2$. Therefore, there exists $c_0>0$ such that
\[
\sum_{j=1}^{3}\mathscr Q^-_{\lambda,\xi,h}(\chi_i\psi)\geq h^{3/2}c_0\min\left(1,\left| \frac{\xi-\delta}{\sqrt{h}}+a_0\right|^2, \left| -\frac{\xi+\delta}{\sqrt{h}}+a_0\right|^2\right)\|\psi\|^2\,.
\]
We obtain then that
\begin{equation}\label{eq.compeig-}
\ell^-_1(\sqrt{h}a_0, \xi,h)\geq h^{3/2}c_0\min\left(1,\left| \frac{\xi-\delta}{\sqrt{h}}+a_0\right|^2, \left| -\frac{\xi+\delta}{\sqrt{h}}+a_0\right|^2\right)-Ch^2\,.
\end{equation}
By Proposition \ref{prop.charact},
if $\ell^-_1(\sqrt{h}a_0, \xi,h)\geq 0$, then $\lambda\mapsto\ell^-_1(\lambda, \xi,h)$ is positive on $(0,\sqrt{h}a_0)$ and  $\mu^-_1(\xi,h)\geq \sqrt{h}a_0$. If $\ell^-_1(\sqrt{h}a_0, \xi,h)\leq 0$, then $\lambda\mapsto\ell^-_1(\lambda, \xi,h)$ is negative on $(\sqrt{h}a_0,+\infty)$, $\sqrt{h}a_0-\mu^-_1(\xi,h)\geq 0$ and by \eqref{eq.cont23},
\[
\sqrt{h}a_0(\sqrt{h}a_0-\mu^-_1(\xi,h)) \leq |\ell^-_1(\sqrt{h}a_0, \xi,h)|\leq 
Ch^2\,,
\]
so
\[
\mu^-_1(\xi,h)\geq \sqrt{h}a_0
-Ch^{3/2}\,.
\]
Note that \eqref{eq.compeig-} suggests that the minima of \(\xi \mapsto \mu^-_1(\xi,h)\) occur near \(\xi = \delta - \sqrt{h}a_0\) and \(\xi = -\delta + \sqrt{h}a_0\) (see Figure \ref{fig:dispCurves}). However, we leave the investigation of a more precise determination of the minima's location as an open question.
We now prove that the control term  $Ch^{3/2}$  can, in fact, be replaced by \( \mathscr{O}(h^\infty) \) using Agmon-type estimates.
For any  $\psi \neq 0$  and  $\xi$  such that
\[
h^{3/2} \frac{2 - a_0^2}{2} \|\psi\|^2 \geq \mathscr{Q}^-_{\lambda, \xi, h}(\psi),
\]
an Agmon-type estimate ensures that
\[
\sum_{j=1}^3 \| h (\partial_t \chi_j) \psi \|^2 = \mathscr{O}(h^\infty) \|\psi\|^2,
\]
which leads to
\[
\mathscr{Q}^-_{\lambda, \xi, h}(\psi) \geq h^{3/2} c_0 \min\left(1,
\left| \frac{\xi - \delta}{\sqrt{h}} + a_0 \right|^2,
\left| -\frac{\xi + \delta}{\sqrt{h}} + a_0 \right|^2
\right)\|\psi\|^2 - \mathscr{O}(h^\infty) \|\psi\|^2.
\]
Therefore, if such  $\psi$  and  $\xi$  exist, then
\[
\ell^-_1(\sqrt{h} a_0, \xi, h) \geq \mathscr{O}(h^\infty),
\]
and the arguments provided above ensure that
\[
\mu^-_1(\xi, h) \geq \sqrt{h} a_0 + \mathscr{O}(h^\infty).
\]
This, along with Point \eqref{pt.propdisp-5} of Proposition \ref{prop.curves-}, concludes the proof.

\subsubsection{Proof of Point \eqref{pt.propdisp-5} of Proposition \ref{prop.curves-}}
Let us focus on upper bounds on $\mu^-_1$. Following the notations of \cite[Section 4.4]{BLTRS23}, we denote by $u_{a_0,-a_0}$ the normalized ground state of $\mathscr{Q}^-_{a_0,-a_0,\mathbb{R}_+}$. Its energy is $0$. Let $v\colon t\in(-\delta,\delta)\mapsto u_{a_0,-a_0}\left(\frac{t+\delta}{\sqrt{h}}\right)$. Since $u_{a_0,-a_0}$ belongs to the Schwartz class, the rescaling performed in the previous section ensures that
\[
	\mathscr Q^-_{\sqrt{h}a_0,\delta - \sqrt{h}a_0,h}(v) = h^{3/2}\mathscr Q^-_{a_0, -a_0, \mathbb{R}_+}(u_{a_0,-a_0}) + \mathscr{O}(h^\infty) = \mathscr{O}(h^\infty)\,,
\]
so that
\[
	\ell^-_1(\sqrt{h}a_0,\delta - \sqrt{h}a_0,h)\leq  \mathscr{O}(h^\infty)\,.
\]
By Proposition \ref{prop.charact},
if $\ell^-_1(\sqrt{h}a_0, \delta - \sqrt{h}a_0,h)\geq 0$, then $\mu^-_1( \delta - \sqrt{h}a_0,h)\geq \sqrt{h}a_0$ and by \eqref{eq.cont23},
\[
\sqrt{h}a_0\left(\mu^-_1( \delta - \sqrt{h}a_0,h)- \sqrt{h}a_0\right)\leq \ell^-_1(\sqrt{h}a_0,\delta - \sqrt{h}a_0,h)\leq  \mathscr{O}(h^\infty)\,,
\]
so that
\[
\mu^-_1( \delta - \sqrt{h}a_0,h)\leq  \sqrt{h}a_0 + \mathscr{O}(h^\infty)\,.
\]
If $\ell^-_1(\sqrt{h}a_0, \delta - \sqrt{h}a_0,h)\leq 0$, then $\mu^-_1( \delta - \sqrt{h}a_0,h)\leq  \sqrt{h}a_0$\,.

\section{Estimates of the discrete spectrum}\label{sec.3}

\subsection{\texorpdfstring{Upper bound of Point \eqref{it.AsymEff+} of Theorem \ref{thm.boundstatesDirac}}{Upper bound}}\label{sec.31}
The key to understanding the asymptotics of the quotients appearing in \eqref{it.AsymEff+} lies in applying Laplace’s method to $\|e^{-\phi/h} u\|$ for fixed $u$. This asymptotic behavior is primarily determined by the value of $\phi$ at its minimum, the Hessian matrix of $\phi$ at $z_{\rm min}$, and the germ of $u$ at $z_{\rm min}$. Specifically, to obtain the lowest mode, we choose a function $u$ such that $u(z_{\rm min}) = 1$, in order to maximize the asymptotic behavior of the denominator $\|e^{-\phi/h} u\|$ in the quotient. To minimize the numerator, we then select the unique function $u$ that minimizes the Hilbert norm $\|\cdot\|_{\partial\Omega}$ among all Hardy functions $u$ satisfying $u(z_{\rm min}) = 1$. To obtain the first excited mode, it is useful to follow the same strategy, with Hardy functions for which $z_{\rm min}$ is a non-degenerate root. The technical challenge of the proof then lies in handling the orthogonality properties of these ‘natural’ test functions.

Let $k\in\mathbb{N}^*$ be fixed in this section.

\begin{notation}\label{not.quasiMode}
	~
	\begin{enumerate}[\rm a)]
	\item We denote by $\{v_0,\dots,v_{k-1}\}\subset \mathscr{H}^2(\Omega)$ a $k$-dimensional family of Hardy functions satisfying	 for $j,l\in\{0,\dots,k-1\}$, 
	$
		v_l^{(j)}(z_{\rm min}) = \delta_{jl}
	$, $\delta_{jl}$ being the Kronecker delta. We assume moreover that $v_{k-1}$ is the unique minimizer of $d^{k}_\mathscr{H} = {\rm dist}_{\mathscr{H}^2(\Omega)}(0,\mathbb{X}_k)$ where $\mathbb{X}_k$ is defined before Theorem \ref{thm.boundstatesDirac}.
	\item The familly $(P_n)_{n\in\mathbb{N}}$  is the $N_{\mathscr{B}}$-orthogonal family obtained after a Gram-Schmidt process on $(1,Z,\dots, Z^n,\dots)$ and normalized by $P_n^{(n)}(0) = b_{n,n}=1$ in $P_n(Z) = \sum_{j=0}^{n}\frac{b_{n,j}}{j!}Z^j$.
	\item For $n\in\{0,\dots,k-1\}$, we define $w_n = \sum_{j=0}^{n}b_{n,j}h^{\frac{n-j}{2}}v_j$.
	\item $\mathrm{Tayl}(w)$ is the polynomial part of the Taylor expansion of degree $k-1$ of the function $w$ at $z_{\rm min}$ :
	\[
		\mathrm{Tayl}(w)(z) = \sum_{l=0}^{k-1}\frac{w^{(l)}(z_{\rm min})}{l!}\left(z-z_{\rm min}\right)^l\,.
	\]
	\end{enumerate}
\end{notation}

\begin{lemma} \label{lem.Taylor1}
We have for $n\in\{0,\dots k-1\}$,
\[
\mathrm{Tayl}(w_n)\left(z_{\rm min} + \sqrt{h}z\right)
=h^{\frac{n}{2}} P_n(z)\,.
\]
\end{lemma}
\begin{proof}
Let $l,n\in\{0,\dots k-1\}$, we have
\[
w_n^{(l)}(z_{\rm min}) 
=  \sum_{j=0}^{n}b_{n,j}h^{\frac{n-j}{2}}v_j^{(l)}(z_{\rm min}) 
= \sum_{j=0}^{n}b_{n,j}h^{\frac{n-j}{2}}\delta_{lj}
= b_{n,l}h^{\frac{n-l}{2}}\,,
\]
and
\[
\mathrm{Tayl}(w_n)\left(z_{\rm min} + \sqrt{h}z\right)
=\sum_{l=0}^{k-1}\frac{w_n^{(l)}(z_{\rm min})}{l!}\left(\sqrt{h}z\right)^l
=h^{\frac{n}{2}} P_n(z)\,.
\]
\end{proof}

\begin{lemma}\label{lem.bargNormsup}
Let $c_0,\dots ,c_{k-1}\in\mathbb{C}$ and $w = \sum_{n=0}^{k-1}c_nw_n$.
We have, as $h\rightarrow 0$,
\[
\int_\Omega\left|\sum_{n=0}^{k-1}c_nw_n\right|^2e^{-2(\phi-\phi_{\rm min})/h}{\rm d}x\geq (1 + o(1))h	\sum_{n=0}^{k-1}|h^{n/2}c_n|^2N_{\mathscr{B}}\left(P_n\right)^2\,.
\]
\end{lemma}

\begin{proof}
Let $\eta\in(1/3,1/2)$.
By Taylor's formula, there exists $C>0$ such that for $x\in D(x_{\rm min}, h^\eta)$
\[\begin{split}
	\phi(x)-\phi_{\rm min} 
	\leq \frac{1}{2}{\rm Hess}_{\rm min}\phi(x-x_{\rm min}, x-x_{\rm min}) + h^{3\eta} C
\end{split}\]
and with $x = x_{\rm min}+ \sqrt{h}y$,
\[\begin{split}
	&\int_\Omega|w|^2e^{-2(\phi-\phi_{\rm min})/h}{\rm d}x
	\geq
	\int_{B(x_{\rm min}, h^\eta)}|w|^2e^{-2(\phi-\phi_{\rm min})/h}{\rm d}x
	\\&\geq e^{-2Ch^{3\eta-1}}\int_{B(x_{\rm min}, h^\eta)}|w|^2e^{-{\rm Hess}_{\rm min}\phi(x-x_{\rm min}, x-x_{\rm min})/h}{\rm d}x
	\\&\geq (1+o(1))\int_{B(x_{\rm min}, h^\eta)}|w|^2e^{-{\rm Hess}_{\rm min}\phi(x-x_{\rm min}, x-x_{\rm min})/h}{\rm d}x
	\\&\geq (1+o(1))h\int_{B(0, h^{\eta-1/2})}|w(x_{\rm min} + \sqrt{h}y)|^2e^{-{\rm Hess}_{\rm min}\phi(y, y)}{\rm d}y\,.	
\end{split}\]
By Lemma \ref{lem.Taylor1}, we have for $y\in D(0, h^{\eta-1/2})$,  $n\in\{0,\dots,k-1\}$,
\[|w_n(x_{\rm min} +\sqrt{h}y)-h^{n/2}P_n(y)|\leq C(\sqrt{h}|y|)^k\,.
\]
Thus we have
\[
	w(x_{\rm min} + \sqrt{h}y)
	= {\mathrm{Tayl}}(w)(x_{\rm min} + \sqrt{h}y)
	+
	R_{h,c}(y)\,,
\]
with
\[
	|R_{h,c}(y)|\leq C|y\sqrt{h}|^{k}\left(\sum_{n=0}^{k-1}|c_n|^2\right)^{1/2}\,.
\]
By Young's inequality, we have for $\varepsilon = h^{1/4}>0$, $\tilde c = (\tilde c_n)_n = (h^{n/2}c_n)_n$,
\[\begin{split}
	&h\int_{D(0, h^{\eta-1/2})}|w(x_{\rm min} + \sqrt{h}y)|^2e^{-{\rm Hess}_{\rm min}\phi(y, y)}{\rm d}y
	\\&\geq
	(1-\varepsilon )h	\int_{D(0, h^{\eta-1/2})}|{\rm Tayl}(w)(x_{\rm min} + \sqrt{h}y)|^2e^{-{\rm Hess}_{\rm min}\phi(y, y)}{\rm d}y
	\\&\quad-\varepsilon^{-1} h\int_{D(0, h^{\eta-1/2})}|R_{h,c}(y)|^2e^{-{\rm Hess}_{\rm min}\phi(y, y)}{\rm d}y
		\\&\geq
	(1-\varepsilon )h	\int_{D(0, h^{\eta-1/2})}|\sum_{n=0}^{k-1}\tilde c_nP_n(y)|^2e^{-{\rm Hess}_{\rm min}\phi(y, y)}{\rm d}y
	-\varepsilon^{-1} h^{3/2}C\|\tilde c\|_2^2
	\\&\geq
	(1 + o(1))h	N_{\mathscr{B}}\left(\sum_{n=0}^{k-1}\tilde c_nP_n(y)\right)^2+o(h)\|\tilde c\|_2^2
	\end{split}\]
	By the orthogonality of the $(P_n)_n$ family, we get
	\[\begin{split}
	N_{\mathscr{B}}\left(\sum_{n=0}^{k-1}\tilde c_nP_n(y)\right)^2
	=
	\sum_{n=0}^{k-1}|\tilde c_n|^2N_{\mathscr{B}}\left(P_n\right)^2			
\end{split}\]
and the equivalence of the norms in finite dimensions ensures
\[
h\int_{D(0, h^{\eta-1/2})}|w(x_{\rm min} + \sqrt{h}y)|^2e^{-{\rm Hess}_{\rm min}\phi(y, y)}{\rm d}y
\geq
(1 + o(1))h	\sum_{n=0}^{k-1}|\tilde c_n|^2N_{\mathscr{B}}\left(P_n\right)^2\,.
\]
\end{proof}
 \begin{lemma}\label{lem.HardySup}
 Let $c_0,\dots ,c_{k-1}\in\mathbb{C}$ and $w = \sum_{n=0}^{k-1}c_nw_n$.
We have, as $h\rightarrow 0$,
 \[\begin{split}
 	\|w\|_{\mathscr{H}^2(\Omega)}
	\leq (1+o(1))\sum_{n=0}^{k-1}|c_n|\|v_n\|_{\mathscr{H}^2(\Omega)}\,,
 \end{split}\]
where $o(1)$ is independent of $(c_0, \dots, c_{k-1})$.

 \end{lemma}
 \begin{proof}
By the triangle inequality and the definition of $(w_n)$, we have
 \[\begin{split}
 	&\|w\|_{\mathscr{H}^2(\Omega)}
	\leq \sum_{n=0}^{k-1}|c_n|\|w_n\|_{\mathscr{H}^2(\Omega)}\,,
	\\&\leq \sum_{n=0}^{k-1}|c_n|\sum_{j=0}^{n}|b_{n,j}|h^{\frac{n-j}{2}}\|v_j\|_{\mathscr{H}^2(\Omega)}
	\\&\leq 
	\sum_{n=0}^{k-1}|c_n|\|v_n\|_{\mathscr{H}^2(\Omega)}
	+
	 \sqrt{h}
	 \sum_{n=1}^{k-1}|c_n|\sum_{j=0}^{n-1}|b_{n,j}|\|v_j\|_{\mathscr{H}^2(\Omega)}
	 \\&\leq 
	(1+\sqrt{h}C)\sum_{n=0}^{k-1}|c_n|\|v_n\|_{\mathscr{H}^2(\Omega)}
	\,,
 \end{split}\]
 where $C = \max\{\sum_{j=0}^{m-1}|b_{m,j}|\|v_j\|_{\mathscr{H}^2(\Omega)}/\|v_m\|_{\mathscr{H}^2(\Omega)}\,;1\leq m\leq k-1\}$.
 \end{proof}

  \begin{lemma}\label{lem.MaxSup}
We have, 
 \[\begin{split}
 	\sup_{c\in\mathbb{C}^k\setminus\{0\}}\frac{h\left(\sum_{j=0}^{k-1}|c_j|\|v_j\|_{\mathscr{H}^2(\Omega)}\right)^2}
	{h	\sum_{j=0}^{k-1}|h^{j/2}c_j|^2N_{\mathscr{B}}\left(P_j\right)^2}
	= 
	(1+o(1))h^{(1-k)}\left(\frac{d^{k}_\mathscr{H}}
	{d^{k}_\mathscr{B}}\right)^2
	\,.
 \end{split}\]

 \end{lemma}
 
 \begin{proof}
Considering $c = (0,\dots,0,1)$ we obtain
  \[\begin{split}
 	\sup_{c\in\mathbb{C}^k\setminus\{0\}}\frac{h\left(\sum_{j=0}^{k-1}|c_j|\|v_j\|_{\mathscr{H}^2(\Omega)}\right)^2}
	{h	\sum_{j=0}^{k-1}|h^{j/2}c_j|^2N_{\mathscr{B}}\left(P_j\right)^2}
	\geq 
	\frac{h^{(1-k)}\|v_{k-1}\|_{\mathscr{H}^2(\Omega)}^2}
	{N_{\mathscr{B}}\left(P_{k-1}\right)^2}\,.
 \end{split}\]
 Conversely, for $c = (c_0,\dots, c_{k-1})\in \mathbb{C}\setminus\{0\}$ we have
   \[\begin{split}
 	&\frac{\sum_{j=0}^{k-1}|c_j|\|v_j\|_{\mathscr{H}^2(\Omega)}}
	{\sqrt{\sum_{j=0}^{k-1}|h^{j/2}c_j|^2N_{\mathscr{B}}\left(P_j\right)^2}}
	=
 	\frac{|c_{k-1}|\|v_{k-1}\|_{\mathscr{H}^2(\Omega)} + \sum_{j=0}^{k-2}|c_j|\|v_j\|_{\mathscr{H}^2(\Omega)}}
	{\sqrt{\sum_{j=0}^{k-1}|h^{j/2}c_j|^2N_{\mathscr{B}}\left(P_j\right)^2}}
	\\&\leq
	\sqrt{\frac{h^{(1-k)}\|v_{k-1}\|_{\mathscr{H}^2(\Omega)}^2}
	{N_{\mathscr{B}}\left(P_{k-1}\right)^2}}
	+
	h^{(2-k)/2}\sup_{c\in\mathbb{C}^{k-1}}\frac{\sum_{j=0}^{k-2}|c_j|\|v_j\|_{\mathscr{H}^2(\Omega)}}
	{\sqrt{\sum_{j=0}^{k-2}|c_j|^2N_{\mathscr{B}}\left(P_j\right)^2}}
	\,.
 \end{split}\]
 Note now that $\|v_{k-1}\|_{\mathscr{H}^2(\Omega)} = d^{k}_\mathscr{H}$ and $N_{\mathscr{B}}\left(P_{k-1}\right) = d^{k}_\mathscr{B}$.
 The result follows.
 \end{proof}
 
The upper bound in Point \eqref{it.AsymEff+} of Theorem \ref{thm.boundstatesDirac} follow from Lemmas \ref{lem.bargNormsup}, \ref{lem.HardySup} and \ref{lem.MaxSup}.
\subsection{\texorpdfstring{Lower bound of Point \eqref{it.AsymEff+} of Theorem \ref{thm.boundstatesDirac}}{Lower bound}}\label{sec.32}

\subsubsection{Preliminaries}
The key point in the proof of the lower bound is to construct an orthogonal projection in $\mathscr{H}^2(\Omega)$ onto a space similar to the one used for the upper bound and to obtain estimates for the remaining terms.
The proof of the lower bound closely follows that presented in \cite[Section 3.1.2]{BLTRS23}. 
A notable difference is that polynomials do not belong to the Hardy space $\mathscr{H}^2(\Omega)$. This is addressed through the introduction of the Hardy-Taylor expansion described below.
Note that in \cite{BLTRS23}, the authors adhered to the classical Taylor expansion and overlooked the elegant and natural projection properties of the Hardy-Taylor expansion (see Lemma \ref{lem.tayl2}). This distinction is a noteworthy difference between the two proofs, offering new insight into the arguments presented.

Let $k\in\mathbb{N}^*$ be fixed in this section. 

\begin{notation}\label{not.quasiMode2}
	~
	\begin{enumerate}[\rm a)]
	\item
For \(l \in \{0, \dots, k-1\}\), we define
\[
\mathbb{X}_{k,l} = \{u \in \mathscr{H}^2(\Omega) : \forall j \in \{0, \dots, k-1\}, \; u^{(j)}(z_{\min}) = \delta_{jl}\}.
\]
Then we denote by \(v_l \in\mathbb{X}_{k,l} \subset  \mathscr{H}^2(\Omega)\) the unique minimizer of \(\mathrm{dist}_{\mathscr{H}^2(\Omega)}(0, \mathbb{X}_{k,l})\). Note that \(\mathbb{X}_{k,k-1} = \mathbb{X}_k\) as defined in Point \eqref{it.AsymEff+} of Theorem \ref{thm.boundstatesDirac}, ensuring that the notation for \(v_k\) is consistent with that introduced in Notation \ref{not.quasiMode}.

	\item 
	We denote by $\mathrm{Tayl}_{\mathscr{H}^2(\Omega)}(w)$ the Hardy - Taylor expansion of degree $k-1$ of the function $w\in\mathscr{H}^2(\Omega)$ at $z_{\rm min}$ :
	\[
		\mathrm{Tayl}_{\mathscr{H}^2(\Omega)}(w) = \sum_{l=0}^{k-1}w^{(l)}(z_{\rm min})v_l\in \mathscr{H}^2(\Omega)\,.
	\]
	\item We denote by $\mathscr{W}_k(h)\subset \mathscr{H}^2(\Omega)$ a subspace such that $\dim \mathscr{W}_k(h) = k$ and
	\begin{equation}\label{eq:sup}
	\sup_{w\in \mathscr{W}_k(h)\setminus\{0\}}\frac{h\|w\|^2_{\mathscr{H}^2(\Omega)}}{\|e^{-\phi/h}w\|^2}\leq \lambda^{\mathrm{eff}}_k(h)(1+o(1))\,.
	\end{equation}
	\end{enumerate}
	
\end{notation}

Let us clarify the relationship between the Taylor expansion and the Hardy-Taylor expansion as defined in Notations \ref{not.quasiMode} and \ref{not.quasiMode2} and prove a notable property of $\mathrm{Tayl}_{\mathscr{H}^2(\Omega)}$.

\begin{lemma}\label{lem.tayl2}
The following assertions hold:
\begin{enumerate}[\rm (i)]
\item \label{pt.orthoproj}The operator
\[
	\mathrm{Tayl}_{\mathscr{H}^2(\Omega)}\colon \mathscr{H}^2(\Omega)\to \mathscr{H}^2(\Omega)
\]
is the $\mathscr{H}^2(\Omega)$-orthogonal projection onto the orthogonal of
\[
\widetilde{\mathbb{X}}_{k+1} = \{u \in \mathscr{H}^2(\Omega) : \forall j \in \{0, \dots, k-1\}, \; u^{(j)}(z_{\min}) =0\}.
\]
\item \label{pt.taylortaylor}
For any \(w \in \mathscr{H}^2(\Omega)\), the following identity holds:
\[
\mathrm{Tayl}\left(w - \mathrm{Tayl}_{\mathscr{H}^2(\Omega)}(w)\right) = 0\,.
\]
\end{enumerate}
\end{lemma}
\begin{proof}
	Let us begin with Point \eqref{pt.taylortaylor}: 
	\[\begin{split}
		\mathrm{Tayl}\left(\mathrm{Tayl}_{\mathscr{H}^2(\Omega)}(w)\right)(z) 
		&= \sum_{l=0}^{k-1}w^{(l)}(z_{\rm min})\mathrm{Tayl}(v_l)(z)
		=
		\sum_{l=0}^{k-1}\frac{w^{(l)}(z_{\rm min})}{l!}(z-z_{\rm min})^l	
		\\&=
		\mathrm{Tayl}(w)(z) 
		\,.	
	\end{split}\]
	
	We now prove Point \eqref{pt.orthoproj}. Note that, by the Cauchy formula estimation from Point \eqref{it.CauchyCurved} of Proposition \ref{prop:HardyCurved}, $\widetilde{\mathbb{X}}_{k+1}$ and $(\mathbb{X}_{k,l})_{0\leq l\leq k-1}$ are (non empty) closed subspaces so that the orthogonal projection is indeed well-defined.
	Let $0\leq l\leq k-1$ and $u\in \widetilde{\mathbb{X}}_{k+1}$, then $u+v_l\in \mathbb{X}_{k,l}$. Since $v_l$ is the orthogonal projection of $0$ onto $\mathbb{X}_{k,l}$, we get
	\[
	0 = \langle v_l, (u+v_l)-v_l \rangle_{\mathscr{H}^2(\Omega)} =  \langle v_l, u\rangle_{\mathscr{H}^2(\Omega)}\,,
	\]
	and $v_l\in (\widetilde{\mathbb{X}}_{k+1})^\perp$. This proves that
	\begin{equation*}
		{\rm span}(v_0,\dots,v_{k-1})\subset (\widetilde{\mathbb{X}}_{k+1})^\perp\,.
	\end{equation*}
	Assume now that $u\in (\widetilde{\mathbb{X}}_{k+1})^\perp\cap 
	({\rm span}(v_0,\dots,v_{k-1}))^\perp$. Then, $u-\mathrm{Tayl}_{\mathscr{H}^2(\Omega)}(u)\in \widetilde{\mathbb{X}}_{k+1}$ so that
	\[
	0 = \langle u, u-\mathrm{Tayl}_{\mathscr{H}^2(\Omega)}(u) \rangle_{\mathscr{H}^2(\Omega)} =  \|u\|^2_{\mathscr{H}^2(\Omega)}\,.
	\]	
	We proved that
	\begin{equation}\label{eq.spanproj2}
		{\rm span}(v_0,\dots,v_{k-1})=(\widetilde{\mathbb{X}}_{k+1})^\perp\,.
	\end{equation}
	Finally, note that 
	\begin{equation}\label{eq.spanproj3}
		\mathrm{Tayl}_{\mathscr{H}^2(\Omega)}(v) 
		= \begin{cases}
		v&\text{ if } v\in {\rm span}(v_0,\dots,v_{k-1})\,,
		\\
		0&\text{ if } v\in  \widetilde{\mathbb{X}}_{k+1}\,.
		\end{cases}
	\end{equation}
	We conclude from \eqref{eq.spanproj2} and \eqref{eq.spanproj3} that $\mathrm{Tayl}_{\mathscr{H}^2(\Omega)}$ is the orthogonal projection onto $(\widetilde{\mathbb{X}}_{k+1})^\perp$ and Point \eqref{pt.orthoproj} follows.
	
		Remark that by the Cauchy formula estimation from Point \eqref{it.CauchyCurved} of Proposition \ref{prop:HardyCurved}, we recover the boundedness of $\mathrm{Tayl}_{\mathscr{H}^2(\Omega)}:$
	\[
\|\mathrm{Tayl}_{\mathscr{H}^2(\Omega)}(w)\|_{\mathscr{H}^2(\Omega)}
\leq 
\sum_{l=0}^{k-1}|w^{(l)}(z_{\rm min})|\|v_l\|_{\mathscr{H}^2(\Omega)}
\leq C \|w\|_{\mathscr{H}^2(\Omega)}\,.
\]
\end{proof}

The next two lemmas provide a priori bounds on the functions in $\mathscr{W}_{k}(h)$.
\begin{lemma}\label{lem.vhL2}
	There exist constants \( C \) and \( h_0 > 0 \) such that for any \( v \in \mathscr{W}_k(h) \) and \( h \) in the range \( (0, h_0) \), the following inequality holds:
\[
\|v\|^2 \leq C h^{-k} e^{2\phi_{\min}/h} \int_{\Omega} e^{-2\phi/h} |v|^2 \, dx.
\]
\end{lemma}

\begin{proof}
From the continuous embedding \( \mathscr{H}^2(\Omega) \hookrightarrow L^2(\Omega) \) of Point \eqref{it.inclL2Curved} of Proposition \ref{prop:HardyCurved} and the upper bound of Point \eqref{it.AsymEff+} of Theorem \ref{thm.boundstatesDirac}, there exist \( c, C, h_0 > 0 \) such that, for all \( h \in (0, h_0) \) and all \( v \in \mathscr{W}_k(h) \),
\[
\begin{split}
ch\|v\|^2 & \leq h\|v\|^2_{\mathscr{H}^2(\Omega)} \leq (1+o(1))\lambda^{\mathrm{eff}}_k(h)\int_{\Omega}e^{-2\phi/h}|v|^2\,dx \\
& \leq C h^{1-k} e^{2\phi_{\min}/h} \int_{\Omega} e^{-2\phi/h}|v|^2\,dx.
\end{split}
\]
\end{proof}
The following lemma comes from \cite[Lemma 3.9]{BLTRS23}.
\begin{lemma}\label{lem.normL2}
    Let \(\alpha \in (1/3, 1/2)\).
    Then,
    \[
    \lim_{h \to 0} \sup_{v \in \mathscr{W}_{k}(h) \setminus \{0\}} \left| \frac{\int_{D(x_{\text{min}}, h^\alpha)} e^{-2\phi/h} |v(x)|^2 \, dx}{\int_{\Omega} e^{-2\phi/h} |v(x)|^2 \, dx} - 1 \right| = 0.
    \]
\end{lemma}

\begin{proof}
Assume that $\alpha\in\left(\frac{1}{3},\frac{1}{2}\right)$. For all $x\in D(x_{\min},h^\alpha)$, we find that
\[\phi(x)=\phi_{\min}+\frac{1}{2}\mathsf{Hess}_{x_{\min}}\phi (x-x_{\min}, x-x_{\min})+\mathscr{O}(h^{3\alpha})\,.\]
According to the maximum principle,
\[
\min_{x\in D(x_{\rm min}, \ h^{\alpha})^c}\phi(x) 
= \min_{x\in \partial D(x_{\rm min}, \ h^{\alpha})^c}\phi(x)
\geq \phi_{\min} + \frac {h^{2\alpha}}2 \min \mathrm{sp}(\mathsf{Hess}_{x_{\min}}) + \mathscr{O}(h^{3\alpha})\,.
\]
Then for any $v \in \mathscr W_k(h)$ we have by Lemma \ref{lem.vhL2}
\[
\frac {\int_{\Omega \setminus D(x_{\min},h^\alpha)} e^{-2\phi/h} |v(x)|^2 \, dx }{\int_{\Omega} e^{-2\phi/h} |v(x)|^2 \, dx} 
\leq \frac {\|v\|^2 e^{-2\phi_{\min}/h} e^{- h^{2\alpha-1} \min \mathrm{sp}(\mathsf{Hess}_{x_{\min}}) + \mathscr O(h^{3\alpha-1})} }{C^{-1} \|v\|^2 h^k e^{-2\phi_{\min}/h} },
\]
and the conclusion follows.
\end{proof}
\subsubsection{Proof of the lower bound}
We are now well-positioned to analyze the lower bound. 

Let  $\alpha\in(1/3,1/2)$ and $v\in\mathscr{W}_k(h)$. With Lemma \ref{lem.normL2},
	\begin{equation}\label{eq.taylphase} 
	he^{2\phi_{\min}/h}\|v\|_{\mathscr{H}^2(\Omega)}^2(1+o(1))
	\leq
	\lambda^{\mathrm{eff}}_k(h)
	\|e^{-\frac{1}{2h}\mathsf{Hess}_{x_{\min}}\phi (x-x_{\min}, x-x_{\min})}v\|_{L^2(D(x_{\min},h^{\alpha}))}^2\,.
	\end{equation}
	In the following, we divide the proof into several parts. First, we replace $v$ with its Taylor expansion of order $k-1$ at $x_{\min}$ in the right-hand side (RHS) of \eqref{eq.taylphase}. Second, we substitute the Hardy-Taylor expansion into the left-hand side (LHS) of the same equation.	
	\begin{enumerate}[\rm i.]
		\item 
		By the Cauchy formula estimation from Point \eqref{it.CauchyCurved} of Proposition \ref{prop:HardyCurved}, there exist constants \(C > 0\) and \(h_0 > 0\) such that for all \(h \in (0, h_0)\), for every \(v \in \mathscr{H}^2(\Omega)\), for all \(z_0 \in D(x_{\min}, h^{\alpha})\), and for each \(n \in \{0, \dots, k\}\),
		\begin{equation}\label{eq.CCS2}
		|v^{(n)}(z_{0})|\leq C\|v\|_{\mathscr{H}^2(\Omega)}\,.
		\end{equation}
		Let us define, for all $v\in \mathscr{H}^2(\Omega)$,
		\[N_{h}(v):=\|e^{-\frac{1}{2h}\mathsf{Hess}_{x_{\min}}\phi (x-x_{\min}, x-x_{\min})}v\|_{L^2(D(x_{\min},h^{\alpha}))}\,.\]
		By the Taylor formula, we can express
\[
v = \mathrm{Tayl}(v) + R(v),
\]
where $\mathrm{Tayl}(v)$ is the $(k-1)$-th degree polynomial Taylor approximation of $v$ at $z_{\min}$, as defined in Notation \ref{not.quasiMode}. Additionally, for all $z_0 \in D(z_{\min}, h^\alpha)$,
		\[|R(v)(z_{0})|\leq C|z_0-z_{\min}|^k\sup_{D(z_{\min}, h^\alpha)}|v^{(k)}|\,.\]
		With \eqref{eq.CCS2} and after rescaling, the Taylor remainder satisfies
\begin{equation}\label{eq.pre35}
N_{h}(R(v)) \leq Ch^{\frac{k+1}{2}} \|v\|_{\mathscr{H}^2(\Omega)}.
\end{equation}
Thus, by the triangle inequality,
\begin{equation}\label{eq.estitaylorBerg}
|N_{h}(v) - N_{h}(\mathrm{Tayl}(v))| \leq Ch^{\frac{k+1}{2}} \|v\|_{\mathscr{H}^2(\Omega)}.
\end{equation}
Therefore, with \eqref{eq.taylphase}, we obtain
\[
(1+o(1))e^{\phi_{\min}/h} \sqrt{h} \|v\|_{\mathscr{H}^2(\Omega)} \leq \sqrt{\lambda^{\mathrm{eff}}_k(h)} N_{h}(\mathrm{Tayl}(v)) + C \sqrt{\lambda^{\mathrm{eff}}_k(h)} h^{\frac{1+k}{2}} \|v\|_{\mathscr{H}^2(\Omega)},
\]
and so, according to the upper bound in Point \eqref{it.AsymEff+} of Theorem \ref{thm.boundstatesDirac},
\begin{equation}\label{eq.ineqTaylor}
(1+o(1))e^{\phi_{\min}/h} \sqrt{h} \|v\|_{\mathscr{H}^2(\Omega)} \leq \sqrt{\lambda^{\mathrm{eff}}_k(h)} N_{h}(\mathrm{Tayl}(v)) \leq \sqrt{\lambda^{\mathrm{eff}}_k(h)} \hat{N}_{h}(\mathrm{Tayl}(v)),
\end{equation}
where
\[
\hat{N}_{h}(w) = \|e^{-\frac{1}{2h}\mathsf{Hess}_{x_{\min}}\phi (x-x_{\min}, x-x_{\min})}w\|_{L^2(\mathbb{R}^2)}.
\]
		By \eqref{eq.estitaylorBerg} and Lemma \ref{lem.tayl2}, we also have
\begin{equation*}
|N_{h}( \mathrm{Tayl}_{\mathscr{H}^2(\Omega)}(v)) - N_{h}(\mathrm{Tayl}(v))| \leq Ch^{\frac{k+1}{2}} \|v\|_{\mathscr{H}^2(\Omega)}
\end{equation*}
so, according to the upper bound in Point \eqref{it.AsymEff+} of Theorem \ref{thm.boundstatesDirac}, and \eqref{eq.ineqTaylor},
\begin{equation}\label{eq.ineqTaylor2}
\begin{split}
(1+o(1))e^{\phi_{\min}/h} \sqrt{h} \|v\|_{\mathscr{H}^2(\Omega)} 
&\leq \sqrt{\lambda^{\mathrm{eff}}_k(h)} N_{h}(\mathrm{Tayl}_{\mathscr{H}^2(\Omega)}(v)) 
\\&\leq \sqrt{\lambda^{\mathrm{eff}}_k(h)} \hat{N}_{h}(\mathrm{Tayl}_{\mathscr{H}^2(\Omega)}(v)).
\end{split}
\end{equation}
		Inequalities \eqref{eq.ineqTaylor} and  \eqref{eq.ineqTaylor2} show that  $\mathrm{Tayl}$ and $\mathrm{Tayl}_{\mathscr{H}^2(\Omega)}$ are injective on $\mathscr{W}_{k}(h)$ and 
		\begin{equation}\label{eq.dim-k}
		\mathrm{dim}\mathrm{Tayl}\mathscr{W}_{k}(h)= \mathrm{dim}\mathrm{Tayl}_{\mathscr{H}^2(\Omega)}\mathscr{W}_{k}(h)=k\,.
		\end{equation}
		
		\item By Lemma \ref{lem.tayl2} and \eqref{eq.ineqTaylor}, we obtain
\begin{equation}\label{eq.ineqTaylor3}
\begin{split}
&\|v\|_{\mathscr{H}^2(\Omega)}^2  = \|\mathrm{Tayl}_{\mathscr{H}^2(\Omega)}(v)\|_{\mathscr{H}^2(\Omega)}^2+ \|v-\mathrm{Tayl}_{\mathscr{H}^2(\Omega)}(v)\|_{\mathscr{H}^2(\Omega)}^2
\\&\leq (1+o(1)) e^{-2\phi_{\min}/h} h^{-1}\lambda^{\mathrm{eff}}_k(h) \hat{N}_{h}(\mathrm{Tayl}(v))^2.
\end{split}
\end{equation}		
Inequality \eqref{eq.ineqTaylor3} ensures that
\begin{equation}\label{eq.quotient5}
\sup_{w\in \mathrm{Tayl}_{\mathscr{H}^2(\Omega)}\mathscr{W}_{k}(h)\setminus \{0\}}\frac{\|w\|_{\mathscr{H}^2(\Omega)}^2}{\hat{N}_{h}(\mathrm{Tayl}(w))^2}
\leq  
(1+o(1)) e^{-2\phi_{\min}/h} h^{-1}\lambda^{\mathrm{eff}}_k(h) \,.
\end{equation}

		\item
		We define $u = \mathrm{Tayl}_{\mathscr{H}^2(\Omega)}(v) - v^{(k-1)}(z_{\rm min})v_{k-1}$. By the triangle inequality, we have
		\begin{equation}\label{eq.est123}\begin{split}
			\|u\|_{\mathscr{H}^2(\Omega)}
			&\leq 
			\sum_{n=0}^{k-2}
			|v^{(n)}(z_{\min})|\|v_n\|_{\mathscr{H}^2(\Omega)}
			\leq 
			h^{-\frac{k-2}{2}}\sum_{n=0}^{k-2}h^{\frac{n}{2}}|v^{(n)}(z_{\min})|\|v_n\|_{\mathscr{H}^2(\Omega)}
			\\
			&\leq h^{-\frac{k-2}{2}}\sum_{n=0}^{k-1}h^{\frac{n}{2}}|v^{(n)}(z_{\min})||\|v_n\|_{\mathscr{H}^2(\Omega)}
			\leq Ch^{-\frac{k-2}{2}} h^{-\frac{1}{2}}\hat N_{h}(\mathrm{Tayl}(v))\,,
		\end{split}\end{equation}
		where we used the rescaling property
		\begin{equation}\label{eq.rescaling.Nh}
		\hat N_{h}\left(\sum_{n=0}^{k-1} c_{n}(z-z_{\min})^n\right)=h^{\frac{1}{2}}\hat N_{1}\left(\sum_{n=0}^{k-1} c_{n}h^{\frac{n}{2}}(z-z_{\min})^n\right)\,,
		\end{equation}
		and the equivalence of the norms in finite dimension: $\exists C>0\,,\forall d\in\mathbb{C}^k$:
		\[ C^{-1}\sum_{n=0}^{k-1}|d_{n}|\|v_n\|_{\mathscr{H}^2(\Omega)}\leq \hat N_{1}\left(\sum_{n=0}^{k-1} \frac{d_{n}}{n!}(z-z_{\min})^n\right) \leq C\sum_{n=0}^{k-1}|d_{n}|\|v_n\|_{\mathscr{H}^2(\Omega)}\,.\]
		Using again the triangle inequality with inequalities \eqref{eq.quotient5},  \eqref{eq.est123} and the upper bound of Point \eqref{it.AsymEff+} of Theorem \ref{thm.boundstatesDirac},
		\[\begin{split}
			&|v^{(k-1)}(z_{\rm min})|d^{k}_\mathscr{H}
			=
			\|v^{(k-1)}(z_{\rm min})v_{k-1}\|_{\mathscr{H}^2(\Omega)} 
			\leq  
			\|u\|_{\mathscr{H}^2(\Omega)} + \|\mathrm{Tayl}_{\mathscr{H}^2(\Omega)}(v)\|_{\mathscr{H}^2(\Omega)}
			\\&
			\leq
			\left(Ch^{-\frac{k-2}{2}} h^{-\frac{1}{2}} 
			+
			(1+o(1)) e^{-\phi_{\min}/h} h^{-1/2}\sqrt{\lambda^{\mathrm{eff}}_k(h)}\right)\hat N_{h}(\mathrm{Tayl}(v))
			\,.
		\end{split}\]
		Remark now that a rescaling ensures that
		\[\begin{split}
			\sup_{p\in \mathbb{C}_{k-1}[X]\setminus\{0\}}\frac{|p^{(k-1)}(z_{\rm min})|}{\hat N_{h}(p)}			
			=
			\frac{1}{h^{k/2}d^{k}_\mathscr{B}}\,.
		\end{split}\]
		By \eqref{eq.dim-k}, ${\rm Tayl} \mathscr{W}_k(h) = \mathbb{C}_{k-1}[X]$,  and 
		\[
		\frac{d^{k}_\mathscr{H}}{d^{k}_\mathscr{B}}h^{-k/2}(1 + o(1))\leq e^{-\phi_{\min}/h} h^{-1/2}\sqrt{\lambda^{\mathrm{eff}}_k(h)}\,.
		\]
		The lower bound follows.
		\end{enumerate}

The proof gives some controls on the functions.
\begin{lemma} 
Let $v=v_h$ be a function that realizes the maximum \eqref{eq:sup}.
There exists $C>0$ such that for all $h\in(0,h_0)$, 
\[\begin{split}
&h\|v-\mathrm{Tayl}_{\mathscr{H}^2(\Omega)}(v)\|_{\mathscr{H}^2(\Omega)}^2
+h\|v^{(k-1)}(z_{\rm min})v_{k-1}-\mathrm{Tayl}_{\mathscr{H}^2(\Omega)}(v)\|_{\mathscr{H}^2(\Omega)}^2
\\&\leq o(\lambda^{\mathrm{eff}}_k(h)) \|e^{-\phi/h}v\|_{L^2(\Omega)}^2\,,
\end{split}\]
\end{lemma}
\begin{proof}
The first inequality comes from the fact that $\mathrm{Tayl}_{\mathscr{H}^2(\Omega)}$ is an orthogonal projection (Lemma \ref{lem.tayl2}) and that $v$ realizes the maximum. The last inequality is a reformulation of \eqref{eq.est123}.
\end{proof}

\subsection{Characterization of the positive eigenvalues and consequences}\label{sec.33}
The main two results in this section are Propositions  \ref{prop.asymptoticmu} \& \ref{prop.characterization}. Combining these two propositions with Point \eqref{it.AsymEff+} of Theorem \ref {thm.boundstatesDirac}, we get Point \eqref{it.AsymEff+2} of Theorem \ref{thm.boundstatesDirac}.

\subsubsection{Statement of the characterization}
First, let us note the following identity:
\begin{equation}\label{eq.idenintert}
	e^{\phi/h}d^\times_\mathbf{A}e^{-\phi/h} = -ih(\partial_1 + i\partial_2) = -2ih\partial_{\overline{z}}.
\end{equation}
 Consider the non-decreasing sequence of numbers 
\begin{equation}\label{eq.mukh}
\begin{split}
&\mu_k(h)
:=\inf_{\underset{\dim W=k}{W\subset\mathfrak{h}_{\mathbf{A}}(\Omega)}} \sup_{u\in W\setminus\{0\}}\frac{h\|u\|^2_{\partial\Omega}+\sqrt{h^2\|u\|^4_{\partial\Omega}+\|d^\times_\mathbf{A}u\|^2\|u\|^2}}{2\|u\|^2}
\\&=
\inf_{\underset{\dim W=k}{W\subset H^1(\Omega)+\mathscr{H}^2(\Omega)}} \sup_{v\in W\setminus\{0\}}
	\frac{h\|v\|^2_{\partial\Omega}+\sqrt{h^2\|v\|^4_{\partial\Omega}+\|e^{-\phi/h}(-2ih\partial_{\overline{z}})v\|^2\|e^{-\phi/h}v\|^2}}{2\|e^{-\phi/h}v\|^2}
\,,
\end{split}\end{equation}
where $\mathfrak{h}_\mathbf{A}(\Omega):=H^1(\Omega)+\mathscr{H}^2_{\mathbf{A}}(\Omega)$ with $\mathscr{H}^2_{\mathbf{A}}(\Omega)=e^{-\phi/h}\mathscr{H}^2(\Omega)$. 
The second equality follows directly from \eqref{eq.mukh}.

Due to Lemma \ref{lem.density}, we also have
\[
\begin{split}
&\mu_k(h)=\inf_{\underset{\dim W=k}{W\subset H^1(\Omega)}} \sup_{u\in W\setminus\{0\}}\frac{h\|u\|^2_{\partial\Omega}+\sqrt{h^2\|u\|^4_{\partial\Omega}+\|d^\times_\mathbf{A}u\|^2\|u\|^2}}{2\|u\|^2}
\\&=
\inf_{\underset{\dim W=k}{W\subset H^1(\Omega)}} \sup_{v\in W\setminus\{0\}}
	\frac{h\|v\|^2_{\partial\Omega}+\sqrt{h^2\|v\|^4_{\partial\Omega}+\|e^{-\phi/h}(-2ih\partial_{\overline{z}})v\|^2\|e^{-\phi/h}v\|^2}}{2\|e^{-\phi/h}v\|^2}
\,.
\end{split}
\]

\begin{proposition}\label{prop.asymptoticmu}
Let $k\geq 1$. Then, we have
\begin{equation}\label{eq.mulambdaeff}
\mu_k(h)=(1+o(1))\lambda_k^{\mathrm{eff}}(h)\,.
\end{equation}
Moreover, for $h$ small enough, $\mathscr{D}_h-\mu_k(h)$ is a Fredholm operator with index $0$.
\end{proposition}
\begin{proof}[Sketch of the proof]
	Note that by choosing test functions for $\mu_k(h)$ in the space $\mathscr{H}^2_{\mathbf{A}}(\Omega)$ and using \eqref{eq.idenintert}, we obtain $0 \leq \mu_k(h) \leq \lambda^{\mathrm{eff}}_k(h)$.
	The lower bound of \eqref{eq.mulambdaeff} is quite similar to the one presented in \cite[Section 3]{BLTRS23}. The key difference is that the domain is unbounded. Nevertheless, the same arguments from \cite[Section 3]{BLTRS23} apply: the test functions can be projected onto the Hardy space to eliminate the term $\|d^\times_\mathbf{A} u\|^2$ in \eqref{eq.mukh}. The necessary controls for this projection are obtained through elliptic estimates, which are valid for both bounded domains  as in \cite{BLTRS23} and unbounded domains with sufficiently flat behavior at infinity, such as half-spaces, straight tubes, and small perturbations of these structures.
	
	The second part of the statement follows from \eqref{eq.mulambdaeff}, Theorem \ref{thm.ess} \& Point \eqref{it.AsymEff+} of Theorem  \ref{thm.boundstatesDirac}. 
	\end{proof}

Let us now state the proposition that  connects the low-lying positive discrete spectrum of $\mathscr{D}_h$ to the $\mu_k(h)$, when $h$ is small.
\begin{proposition}\label{prop.characterization}
	Let $k\geq 1$. There exists $h_0>0$ such that, for all $h\in(0,h_0)$, the $k$-th positive eigenvalue of $\mathscr{D}_h$ exists and satisfies
\[\lambda^+_k(h)=\mu_k(h)\,.\]	
\end{proposition}
To prove Proposition \ref{prop.characterization}, we revisit in the following the proof of \cite[Proposition 3.2]{BLTRS23}. Here, we emphasize the differences caused by the unboundedness of our domain $\Omega$ compared to the bounded case considered in \cite{BLTRS23}. The key idea is that, although the operators do not have compact resolvents, everything works similarly to the bounded case as long as the considered eigenvalues are not embedded in the essential spectrum.

\subsubsection{Characterization of the $\mu_k(h)$ and relation to $\mathscr{D}_h$}
Let $\lambda>0$. For all $u\in \mathfrak{h}_{\mathbf{A}}(\Omega)$, we consider
\[q_\lambda(u)=\|d^\times_{\mathbf{A}}u\|^2+h\lambda\|u\|^2_{\partial\Omega}-\lambda^2\|u\|^2\,.\]
\begin{lemma}
The quadratic form $q_\lambda$ is closed on its domain $\mathfrak{h}_{\mathbf{A}}(\Omega)$.
\end{lemma}
\begin{proof}[Sketch of the proof]
The proof is similar to the case when $\Omega$ is bounded, see \cite[Lemma 2.4 \& Proposition 2.5 \rm (i)]{BLTRS23}, since the arguments do not use the boundedness. 

Let us outline the proof. Consider a sequence $(u_n)$ converging for $q_\lambda$. By projecting the sequence $(u_n)$ onto the orthogonal complement of the kernel of $d^\times_{\mathbf{A}}$, we obtain a sequence $(v_n)$ whose elements lie in the range of the adjoint of $d^\times_{\mathbf{A}}$. By elliptic estimates, this sequence converges in the $H^1$-norm, as well as in $\|\cdot\|_{\partial\Omega}$. Consequently, the sequence $(e^{\phi/h}(u_n - v_n)) \subset \mathscr{H}^2(\Omega)$ also converges in $\mathscr{H}^2(\Omega)$, and thus in $\|\cdot\|$.
	\end{proof}
We denote by $\mathscr{L}_\lambda$ the self-adjoint operator associated with $q_\lambda$. We denote by $(\ell_{k}(\lambda))_{k\geq 1}$ the non-decreasing sequence of the Rayleigh quotients of $q_\lambda$.

Let us explain the relations between $\mathscr{L}_\lambda$ and $\mathscr{D}_h-\lambda$.

\begin{proposition}\label {prop.J}
For all $u\in \mathrm{Dom}\,\mathscr{L}_\lambda$, we let
	\[\mathscr{J}_\lambda(u)=(u,\lambda^{-1}d^\times_{\mathbf{A}}u)\,.\]
	Then, we have the following
	\begin{enumerate}[\rm (i)]
	\item\label{eq.Ji} The application $\mathscr{J}_\lambda$ sends $\mathrm{Dom}\,\mathscr{L}_\lambda$ into $\mathrm{Dom}\,\mathscr{D}_h$ and, for all $u\in \mathrm{Dom}\,\mathscr{L}_\lambda$, we have
	\[(\mathscr{D}_h-\lambda)\mathscr{J}_\lambda(u)=(\mathscr{L}_\lambda u,0)\,.\]
		\item\label{eq.Jii} The application $\mathscr{J}_\lambda$ induces an isomorphism from $\ker\mathscr{L}_\lambda$ to $\ker(\mathscr{D}_h-\lambda)$. 
	\item \label{eq.Jiii}The application $\mathscr{J}_\lambda$ has closed range in $\mathrm{Dom}\,\mathscr{D}_h$.

	\item\label{eq.Jiv} If $\mathscr{D}_h-\lambda$ is Fredholm with index $0$, then so is $\mathscr{L}_\lambda$. In particular, if $\lambda\in\mathrm{sp}_{\mathrm{dis}}(\mathscr{D}_h)$, then $0\in\mathrm{sp}_{\mathrm{dis}}(\mathscr{L}_\lambda)$.
	\end{enumerate}
\end{proposition}
\begin{proof}
Take $u\in \mathrm{Dom}\,\mathscr{L}_\lambda\subset \mathrm{Dom}\, q_\lambda$. Then, for all $v =(v_1,v_2) \in\mathrm{Dom}\,\mathscr{D}_h$,
\[\begin{split}\langle \mathscr{J}_\lambda(u), (\mathscr{D}_h-\lambda)v\rangle&=\langle u,-\lambda v_1+d_{\mathbf{A}} v_2\rangle+\langle \lambda^{-1}d^\times_{\mathbf {A}}u,-\lambda v_2+d_{\mathbf{A}}^\times v_1\rangle\\
&=\lambda^{-1}q_\lambda(u,v_1)\,,
\end{split}\]
where we used an integration by parts and the boundary condition satisfied by $v$. Since $u\in\mathrm{Dom}\,\mathscr{L}_\lambda$, we get that
\[q_\lambda(u,v_1)=\langle \mathscr{L}_\lambda u,v_1\rangle\,.\]
This shows that $\mathscr{J}_\lambda(u)\in\mathrm {Dom}\,\mathscr{D}_h^*=\mathrm {Dom}\,\mathscr{D}_h$ and Point \eqref{eq.Ji} follows.

Thanks to \eqref{eq.Ji}, we have only to check that 
$\mathscr{J}_\lambda : \ker\mathscr{L}_\lambda\to \ker(\mathscr{D}_h-\lambda)$ is surjective (since it is clearly injective). Take $v\in\ker(\mathscr{D}_h-\lambda)$. We have $d_{\mathbf{A}}v_2=\lambda v_1$ and $d^\times_{\mathbf{A}}v_1=\lambda v_2$. Let us check that $v_1\in\ker\mathscr{L}_\lambda$. For all $w \in \mathfrak{h}_{\mathbf{A}}(\Omega)$ we have
\[\begin{split}
	q_\lambda(v_1,w)&=\langle d^\times_{\mathbf{A}}v_1,d^\times_{\mathbf{A}}w\rangle+h\lambda\langle v_1,w\rangle_{\partial\Omega}-\lambda^2\langle v_1,w\rangle\\
	&=\lambda\langle v_2,d^\times_{\mathbf{A}}w\rangle+h\lambda\langle v_1,w\rangle_{\partial\Omega}-\lambda^2\langle v_1,w\rangle=0\,,
	\end{split}\]
where we used an integration by parts and the boundary condition. This proves Point \eqref{eq.Jii}.

Point \eqref{eq.Jiii} follows from the fact that the graph norm of $\mathscr{D}_h$ is the $H^1$-norm. In particular, $\mathscr{J}_\lambda$ is a continuous isomorphism between Banach spaces, from $\mathrm{Dom}\,\mathscr{L}_\lambda$ onto its closed range. Thus, $\mathscr{J}_\lambda$ is a Fredholm operator with index $0$. 

	The identity in \eqref{eq.Ji} implies \eqref{eq.Jiv} because the product of Fredholm operators with index 0 remains a Fredholm operator with index 0. Furthermore, $\lambda \in \mathrm{sp}_{\mathrm{dis}}(\mathscr{D}_h)$ is equivalent to stating that $\lambda$ belongs to the spectrum and that $\mathscr{D}_h - \lambda$ is a Fredholm operator with index 0.
	\end{proof}

\begin{lemma}\label{lem.ell1>0}
For $h$ small enough, we have $\mu_1(h)>0$.  Moreover, for all $\lambda\in(0,\mu_1(h))$, we have $\ell_1(\lambda)>0$.	
	\end{lemma}
\begin{proof}
The first part of the statement follows from  Proposition \ref{prop.asymptoticmu} and Theorem \ref{thm.boundstatesDirac} \eqref{it.AsymEff+}. Then, we take $\lambda\in(0,\mu_1(h))$. For all  $u\in \mathfrak{h}_{\mathbf{A}}(\Omega)$, we have
\[\rho_+(u):=\frac{h\|u\|^2_{\partial\Omega}+\sqrt{h^2\|u\|^4_{\partial\Omega}+\|d^\times_\mathbf{A}u\|^2\|u\|^2}}{2\|u\|^2}\geq\mu_1(h)>0\,.\]
We have, for all $u\in\mathfrak{h}_{\mathbf{A}}(\Omega)$ such that $\|u\|=1$,
\[q_\lambda(u)=-(\lambda-\rho_+(u))(\lambda-\rho_-(u))\,,\]
with $\rho_-(u)\leq 0$.  We have $\rho_+(u)-\lambda\geq\mu_1(h)-\lambda>0$
and
\[
q_\lambda(u)\geq (\mu_1(h)-\lambda)\lambda>0\,.
\]
 The conclusion follows.

	\end{proof}
	
	The following lemma essentially comes from \cite[Proposition 2.11]{BLTRS23} and makes the bridge between the $\mu_k(h)$ and the spectrum of $\mathscr{D}_h$ through $\mathscr{L}_\lambda$.
\begin{lemma}
Let $k\geq 1$. The equation $\ell_k(\lambda)=0$ admits $\mu_k(h)$ as unique positive solution. 
	\end{lemma}
\begin{proof}
The proof of existence and uniqueness follows the same reasoning as in \cite[Lemma 2.10]{BLTRS23}. 
Let us recall some key elements. The existence is guaranteed by the continuity of $\ell_k$, the fact that $\ell_k(\lambda)$ is positive for small $\lambda$, and that $\lim_{\lambda \to +\infty} \ell_k(\lambda) = -\infty$.

The uniqueness is based on \cite[Identity (2.10)]{BLTRS23}: for $0 < \lambda_1 < \lambda_2$, we have
\[
	\lambda_1^{-1}\ell_k(\lambda_1)\geq (\lambda_2-\lambda_1) + \lambda_2^{-1}\ell_k(\lambda_2)\,.
\] 
Let us check that $\mu_k(h)$ solves the equation. Take $\epsilon>0$ and consider $W\subset \mathfrak{h}_{\mathbf{A}}(\Omega)$ with $\dim W=k$ such that
\[\max_{u\in W\setminus\{0\}}\rho_+(u)\leq \mu_k(h)+\epsilon\,.\]
In particular, for all $u\in W$, we have $\rho_+(u)\leq \mu_k(h)+\epsilon$. Then, we have
\[\begin{split}
	\ell_k(\mu_k(h))
	\leq \max_{\underset{\|u\|=1}{u\in W\setminus\{0\}}} 
		q_{\mu_k(h)}(u)
	&=\max_{\underset{\|u\|=1}{u\in W\setminus\{0\}}}
		(\mu_k(h)-\rho_-(u))(\rho_+(u)-\mu_k(h)) \\
	&\leq\epsilon \max_{\underset{\|u\|=1}{u\in W\setminus\{0\}}}
		(\mu_k(h)-\rho_-(u))\,.
\end{split}\]
Taking the limit $\epsilon\to0$, we get	$\ell_k(\mu_k(h))\leq 0$. Conversely, for all $W\subset \mathfrak{h}_{\mathbf{A}}(\Omega)$ such that $\dim W=k$, there exists $u_k\in W\setminus\{0\}$ and $\|u_k\| = 1$ such that
\[\mu_k(h)\leq\max_{u\in W\setminus\{0\}}\rho_+(u)=\rho_+(u_k)\,, \]
and then
\[
\max_{\underset{\|u\|=1}{u\in W\setminus\{0\}}} 
		q_{\mu_k(h)}(u)
\geq 
	q_{\mu_k(h)}(u_k)
=
	(\mu_k(h)-\rho_-(u_k))(\rho_+(u_k)-\mu_k(h))
\geq 0\,.
\]
We used here the fact that the second root $\mu_-(u_k)$ is non positive.
Taking the infimum over $W$, we get $\ell_k(\mu_k(h))\geq 0$.
\end{proof}

\subsubsection{Proof of Proposition \ref{prop.characterization}}\label{sec.proof312}

For all $k\geq 1$ and for $h$ small enough $\mathscr{D}_h-\mu_k(h)$ is Fredholm with index $0$. From Proposition \ref{prop.J}, we get that $\mathscr{L}_{\mu_k(h)}$ is Fredholm with index $0$ and with a non-empty kernel (of finite dimension) since $\ell_k(\mu_k(h))=0$. We get that $\mu_k(h)\in\mathrm{sp}_{\mathrm{dis}}(\mathscr{D}_h)$. This shows that, for all $k\geq 1$,
\begin{equation}\label{eq.ublamu}
\lambda^+_k(h)\leq \mu_k(h)\,.
\end {equation}
Let us explain this. Assume that $\mu_1(h)=\ldots=\mu_{k_1}(h)<\mu_{k_1+1}(h)$. This implies that $\ell_1(\mu_1(h))=\ell_{k_1}(\mu_1(h))=0$ and $\ell_{k_1+1}(\mu_1(h))>0$ so that, by the min-max theorem, $\dim\ker\mathscr{L}_{\mu_1(h)}= k_1$ and then $\dim\ker(\mathscr{D}_{h}-\mu_1(h))= k_1$.  This shows \eqref{eq.ublamu} for $k=1,\ldots,k_1$. By induction and similar considerations, we get  \eqref{eq.ublamu}. 

Conversely, we notice that $\mathscr{L}_{\lambda_k^+(h)}$ is Fredholm with index $0$, with non-empty kernel so that, for some $p$, $\ell_{p}(\lambda_k^+(h))=0$. Thus, for some $p$, $\lambda_k^+(h)=\mu_p(h)$. Moreover, assume that $\lambda_1^+(h)$ has multiplicity $m_1$. Thus, we have $\dim\ker\mathscr{L}_{\lambda_1^+(h)}=m_1$. Therefore there exists $p\in\mathbb{N}$ such that $\ell_{p+1}(\lambda_1^+(h))=\ldots=\ell_{p+m_1}(\lambda_1^+(h))=0$ and $\ell_p(\lambda_1^+(h))<0<\ell_{p+m_1+1}(\lambda_1^+(h))$. In particular, $\mu_{p+1}(h)=\ldots=\mu_{p+m_1}(h)=\lambda_1^+(h)$. This shows that $\mu_k(h)\leq\lambda_k^+(h)$ for $k=1,\ldots,m_1$. By induction, we can check that this inequality is true for $k\geq 1$.

\section*{Acknowledgments}
This work was conducted within the France 2030 framework programme, Centre Henri Lebesgue ANR-11-LABX-0020-01 and CIMI ANR-11-LABX-0040. This work has been partially supported by CNRS International Research Project Spectral Analysis of Dirac Operators – SPEDO. The authors are grateful to CIRM where this work was started and they also thank  Enguerrand Lavigne-Bon for many discussions at the origin of this work.

\appendix
\section{Hardy space on the strip}\label{app.A}
\subsection{Hardy space on the straight strip}
Let us consider the strip $S_\delta=\R\times(-\delta,\delta)$ and consider the following set of holomorphic functions
\[
\begin{split}
	\mathscr{H}^2(S_\delta)&= \mathscr{O}(S_\delta)\cap L^\infty((-\delta,\delta)_y, L^2(\mathbb{R}_x))
	\\&= \{u\in\mathscr{O}(S_\delta) : M(u):=\sup_{y\in(-\delta,\delta)}\|u(\cdot+iy)\|_{L^2(\R)}<+\infty\}\,.
\end{split}
\]
Let us gather the well-known properties of the Hardy space $\mathscr{H}^2(S_\delta)$ (see, for instance, \cite[Chapter 19]{Rudin} dealing with the half-space).
\begin{proposition}\label{prop:Hardy}
The following holds.
\begin{enumerate}[\rm (i)]
\item \label{it.Banach}
The space $\mathscr{H}^2(S_\delta)$ is a Banach space.
\item \label{it.PaleyWiener} [Paley-Wiener] For all $u\in \mathscr{H}^2(S_\delta)$, the map
\[(-\delta,\delta)\ni y\mapsto u(\cdot+iy)\in L^2(\R)\]
is continuous and can be extended by continuity to $[-\delta,\delta]$. This defines a trace operator at the boundary :
\[
	T\colon u\in \mathscr{H}^2(S_\delta)\mapsto T(u) \in L^2(\partial S_\delta)\,.
\]
\item \label{it.Equiv} The norms $u\mapsto \|Tu\|_{L^2(\partial S_\delta)} =  \sqrt{\|u(\cdot-i\delta)\|^2_{L^2(\R)}+\|u(\cdot+i\delta)\|^2_{L^2(\R)}}$ and $u\mapsto M(u)$ are equivalent. Moreover, $\mathscr{H}^2(S_\delta)$ endowed with $\|T\cdot\|_{L^2(\partial S_\delta)}$ is a Hilbert space and $T$ becomes an isometry. 
	\item\label{it.inclL2} 
	We have the continuous embedding $\mathscr{H}^2(S_\delta)\subset L^2(S_\delta)\,,$
with 
$\|u\|_{L^2(S_\delta)}\leq \sqrt{\delta }\|Tu\|_{L^2(\partial S_\delta)}$, 
for all $u\in \mathscr{H}^2(S_\delta)$.
\item \label{it.Cauchy} For $u\in \mathscr{H}^2(S_\delta)$, $z_0\in S_\delta$ and $k\in\mathbb{N}$, we have
\[\begin{split}
	|u^{(k)}(z_0)|  
		&\leq
	\sqrt{\frac{(2k)!}{2^{2k+1}\pi}}
	{\rm dist}(z_0,\partial S_\delta)^{-\frac{2k+1}{2}}
	\|
	Tu
	\|_{L^2(\partial S_\delta)}\,.
\end{split}\]

\end{enumerate}
\end{proposition}
\begin{proof}
$L^\infty((-\delta,\delta)_y, L^2(\mathbb{R}_x))$ is a Banach space  that is continuously embedded in $L^1_{\rm loc}(S_\delta)$. Therefore, the distribution theory ensures that $\mathscr{H}^2(S_\delta)$ is a closed subset and Point \eqref{it.Banach} follows.
To show Point \eqref{it.PaleyWiener}, consider $u\in \mathscr{H}^2(S_\delta)$ and $y\in(-\delta,\delta)$. We can consider the partial Fourier transform $\mathscr{F}u(\cdot+iy)\in L^2(\R)$ and check by Cauchy formula that
\[\mathscr{F}[u(\cdot+iy)](\xi)=e^{-y\xi}\mathscr{F}[u(\cdot)](\xi)\,.\]
From the Parseval formula, it follows that, for all $y\in(-\delta,\delta)$,
\[\int_{\R}e^{-2y\xi}|\mathscr{F}[u(\cdot)](\xi)|^2\dd\xi=\|u(\cdot+iy)\|^2_{L^2(\R)}\leq M(u)^2\,.\]
Thanks to the Fatou lemma (by sending $y\to \pm\delta$), we see that
\[\int_{\R}\cosh(2\delta\xi)|\mathscr{F}[u(\cdot)](\xi)|^2\dd\xi\leq M(u)^2\,,\]
and in particular that
\[M(u)^2\leq\int_{\R}e^{2\delta|\xi|}|\mathscr{F}[u(\cdot)](\xi)|^2\dd\xi\leq 2M(u)^2\,.\]
By using the dominated convergence theorem, it shows that  $(-\delta,\delta)\ni y\mapsto u(\cdot+iy)\in L^2(\R)$ is continuous. This application has also limits in $\pm\delta$. These limits are $\mathscr{F}^{-1}(e^{\mp \delta\xi}\mathscr{F}[u(\cdot)])$.
Point \eqref{it.Equiv} follows from Point \eqref{it.PaleyWiener}.
Let us turn to Point \eqref{it.inclL2}. Let $u\in \mathscr{H}^2(S_\delta)$, we have
\[\begin{split}
	\|u\|_{L^2(S_\delta)}^2
	&=
	\int_{-\delta}^\delta\|\mathscr{F}u(\cdot+iy)\|^2_{L^2(\R)}{\rm d}y
	=\delta\int_\mathbb{R}|\mathscr{F}u|^2\left(\frac{e^{2\xi\delta}-e^{-2\xi\delta}}{2\xi\delta}\right){\rm d}\xi
	\\&\leq
	\delta\|T u \|_{\partial S_\delta}^2\left\|\frac{\tanh(x)}{x}\right\|_{L^\infty}
	\leq \delta\|T u \|_{\partial S_\delta}^2\,.
\end{split}\]
Let us now show Point \eqref{it.Cauchy}. Let $u\in \mathscr{H}^2(S_\delta)$ and $z_0 = x_0 + iy_0\in S_\delta$. We have
\[\begin{split}
	(-i)^ku^{(k)}(z_0)  &= (-i)^k\partial_x^ku(z_0)= \mathscr{F}^{-1}\left(\xi^k\mathscr{F}(u(\cdot+iy_0))\right)(x_0)
	\\&=\mathscr{F}^{-1}\left(
	\xi^ke^{-y_0\xi}
	\mathscr{F}[u(\cdot)]
	\right)(x_0)\,,
\end{split}\]
so that the Cauchy-Schwarz inequality ensures 
\[\begin{split}
	|u^{(k)}(z_0)|  
	&\leq \frac{1}{\sqrt{2\pi}}\|
	\xi^ke^{-y_0\xi}
	\mathscr{F}[u(\cdot)]
	\|_{L^1(\mathbb{R})}
	\\&\leq
	\frac{1}{\sqrt{2\pi}}
	\|
	|\xi|^ke^{-(\delta-|y_0|)|\xi|}
	\|_{L^2(\mathbb{R})}
	\|
	e^{\delta|\xi|}
	\mathscr{F}[u(\cdot)]
	\|_{L^2(\mathbb{R})}
	\\&\leq
	\frac{1}{\sqrt{2\pi}(\delta-|y_0|)^{\frac{2k+1}{2}}}
	\|
	|\xi|^ke^{-|\xi|}
	\|_{L^2(\mathbb{R})}
	\|
	Tu
	\|_{L^2(\partial S_\delta)}\,.
\end{split}\]

\end{proof}

\begin{lemma}\label{lem.density0}
	The space $H^1(S_\delta)\cap\mathscr{H}^2(S_\delta)$ is dense in $\mathscr{H}^2(S_\delta)$. More precisely, for all $u\in \mathscr{H}^2(S_\delta)$, there exists $(u_\epsilon)_{\epsilon>0}\subset H^1(S_\delta)\cap\mathscr{H}^2(S_\delta)$ such that \[\lim_{\epsilon\to +0}\|T(u_\epsilon-u)\|_{L^2(\partial S_\delta)}=0\,.\]
\end{lemma}
\begin{proof}
 Let $u\in\mathscr{H}^2(S_\delta)$	and $\epsilon>0$. We let
	\[u_\epsilon(x)=u((1-\epsilon)x)\,.\]
	The function $u_\epsilon$ belongs to $\mathscr{O}(S_{\delta/(1-\epsilon)})$. In particular, $u_\epsilon\in\mathscr{C}^\infty(\overline{S}_\delta)$. We also see that $u_\epsilon\in\mathscr{H}^2(S_\delta)$. In fact, $u_\epsilon\in H^1(S_\delta)$. To see this, we notice that
	\[\mathscr{F}(u_\epsilon(\cdot+iy))=e^{-y\xi}\mathscr{F}(u_\epsilon)=(1-\epsilon)^{-1}e^{-y\xi}\mathscr{F}(u)((1-\epsilon)^{-1}\xi)\,.
	\]
	We have
	\[\begin{split}\int_{\R} e^{\alpha|\xi|}|\mathscr{F}(u_\epsilon(\cdot+iy))|^2\mathrm{d}\xi&=(1-\epsilon)^{-1}\int_{\R}e^{\alpha|\xi|}e^{-2y(1-\epsilon)\xi}|\mathscr{F}[u]|^2\mathrm{d}\xi\\
		&\leq(1-\epsilon)^{-1}\int_{\R}e^{(\alpha-2\delta\epsilon)|\xi|}e^{2\delta|\xi|}|\mathscr{F}[u]|^2\mathrm{d}\xi\,.
	\end{split}\]
	We recall that
	\[\int_{\R}\cosh(2\delta\xi)|\mathscr{F}[u(\cdot)](\xi)|^2\dd\xi\leq M(u)^2\,.\]
	Thus, by taking $\alpha=\delta\epsilon$, we get
	\[\int_{\R} e^{\alpha|\xi|}|\mathscr{F}(u_\epsilon(\cdot+iy))|^2\mathrm{d}\xi\leq 2(1-\epsilon)^{-1}M(u)^2\,.\]
	Integrating then with respect to $y$, we infer that $\partial_x^n u_\epsilon\in L^2(S_\delta)$ for all $n\geq 1$. By using the Cauchy-Riemann relation $\partial_x u_\epsilon+i\partial_y u_\epsilon=0$, we get that $u_\epsilon\in H^1(S_\delta)$.
	
	Let us now consider the approximation. We have
	\[\|T(u_\epsilon-u)\|^2_{L^2(\partial S_\delta)}=\int_{\R}(e^{-2\delta\xi}+e^{2\delta\xi})|\mathscr{F}(u_\epsilon-u)|^2\mathrm{d}\xi\,,
	\]
	which goes to $0$ as $\epsilon$ goes to $0$ by the dominate convergence theorem.
\end{proof}

\subsection{Biholomorphism}
We would like to define the Hardy space $\mathscr{H}^2(\Omega_\delta)$. Of course, by the Riemann mapping theorem, we can transform $\Omega_\delta$ into $S_\delta$ or even into the unit disk by means of a biholomorphism. As we can guess, the problem of defining the Hardy space is the behavior of the biholomorphism near the boundary. 
The purpose of this section is to construct a biholomorphism whose derivatives are well-controlled up to the boundary.

\begin{proposition}\label{prop.biholo}
There exist $\delta_0, C>0$ and for all $\delta\in(0,\delta_0)$, a biholomorphism  $f\colon\Omega_\delta\to S_\delta$  such that
\[
	\|\tilde f(s,t)-(s+it)\|_{\mathscr{C}^1(\overline{S_\delta})}\leq C\delta\,.\]
\end{proposition}
The following propositions will allow the construction of the imaginary part of $f$.
\begin{proposition}
There exists a unique function $\beta\colon \Omega_\delta\to \mathbb{R}$ such that $\tilde\beta(s,t)-t\in H_0^1(S_\delta)$ and satisfying
\[\Delta\beta=0\,,\quad \beta_{|\Gamma_\delta^{\pm}}=\pm\delta\,.\]	
In fact, $\tilde\beta(s,t)-t\in \mathscr{S}(\overline{S_\delta})$ and in particular $\beta\in\mathscr{C}^\infty(\overline{\Omega_\delta})$.
\end{proposition}
\begin{proof}
Define $t\colon \Omega_\delta\to \mathbb{R}$ the transverse coordinate to $\gamma$. Note that $\Delta t\in \mathscr{C}^\infty_c(\Omega_\delta)$ since $\kappa\in \mathscr{C}^\infty_c(\mathbb{R})$. We are led to solve in $H^1_0(\Omega_\delta)$ the Poisson problem
$
	\Delta u = -\Delta t\,.	
$
The unique solution is then $\beta = u + t$. For more details, one refers to \cite{BLLTRR23}.
\end{proof}

Following the same analysis as in \cite[Section 3.2]{BLLTRR23}, we can prove that, when $\delta$ is small enough, $\beta$ is approximated by $t$.
\begin{proposition}\label{prop.tcoordinates}
There exist $\delta_0,C>0$ such that, for all $\delta\in(0,\delta_0)$,
\[\|\tilde\beta(s,t)-t\|_{\mathscr{C}^1(\overline{S_\delta})}\leq C\delta\,.\]	
In particular, $\nabla\beta$ is uniformly non-zero on $\overline{\Omega}_\delta$.
\end{proposition}

We recall Poincaré's Lemma.
\begin{lemma}
	Let $x_0\in \gamma(\mathbb{R})$.
	For all $x\in\overline{\Omega}_\delta$, we let
	\[\alpha(x)=\int_{\gamma_{x_0,x}} (\nabla\beta)^{\perp}\cdot\overrightarrow{\mathrm{d}\ell}\,,\]
	where $\gamma_{x_0,x}$ is a path of class $\mathscr{C}^1$ connecting $x_0$ to $x$. 	
	
	Then, the function $\alpha$ is well defined (it does not depend on the choice of path) and it is a smooth function on $\overline{\Omega}_\delta$ that satisfies
	\[\nabla\alpha=(\nabla\beta)^{\perp}\,.\]
\end{lemma}

We let $f=\alpha+i\beta$. Then, by construction, we see that $f$ is holomorphic on $\Omega_\delta$.
\begin{proposition}\label{prop.tcoordinates2}
There exist $\delta_0,C>0$ such that, for all $\delta\in(0,\delta_0)$,
\[\|\tilde f(s,t)-(s+it)\|_{\mathscr{C}^1(\overline{S_\delta})}\leq C\delta\,.\]	
\end{proposition}

 It remains to show that $f$ is a biholomorphism.


\begin{lemma}
We have
\[f(\Omega_\delta)\subset S_\delta\,.\]
\end{lemma}
\begin{proof}
By Proposition \ref{prop.tcoordinates}, $\beta$ is bounded. Let $(s_n, t_n)_n$ a maximizing sequence for $\beta$. Either there is a bounded subsequence, then the limit is attained at the boundary due to the maximum principle or there is a subsequence such that $|s_n|\to+\infty$ and $t_n\to t_\infty\in[-\delta,\delta]$. Proposition \ref{prop.tcoordinates} ensures that $t_\infty = \sup \beta\in[-\delta,\delta]$ so that $\sup\beta = \delta$. The same holds for the infimum.
\end{proof}

\begin{lemma}\label{lem:surj}
We have $f(\Omega_\delta)=S_\delta$.	
\end{lemma}

\begin{proof}
Since $f$ is not constant, $f(\Omega_\delta)$ is an open set (by the open mapping theorem), which is also connected. Let us show that is it closed in $S_\delta$. Consider a sequence $x_n\in \Omega_\delta$ such that $\lim_{n\to+\infty}f(x_n)=\ell\in S_\delta$.	If $(x_n)$ is not bounded, we may assume that $s_n\to+\infty$ and thus $(f(x_n))$ is not bounded. Thus, $(x_n)$ is bounded and we may assume that $x_n\to x_\infty\in\overline{\Omega}_\delta$. We have $\ell=f(x_\infty)$ since $f$ is continuous. We cannot have $x_\infty\in\partial\Omega_\delta$ since $\ell\in S_\delta$. Therefore $x_\infty\in\Omega_\delta$. By connectedness, we get the result.
\end{proof}

\begin{lemma}\label{lem:inj}
There exists $\delta_0>0$ such that for all $\delta\in(0,\delta_0)$, $f$ is injective. 
\end{lemma}

\begin{proof}
Assume by contradiction that there is a sequence $(\delta_n)\to 0$ such that $f_{\delta_n}$ is not injective. There exist $(x_n^1),(x_n^2)\subset \Omega_{\delta_n}$ such that $x_n^1\ne x_n^2$ and $f_{\delta_n}(x_n^1) = f_{\delta_n}(x_n^2)$. By the Taylor formula and Proposition \ref{prop.tcoordinates2}, 
\[\begin{split}
 0 &= \tilde f_{\delta_n}(y_n^2) - \tilde f_{\delta_n}(y_n^1)
 =\int_0^1d\tilde f_{\delta_n}(y_n^1 + u(y^2_n-y^1_n))\cdot (y^2_n-y^1_n){\rm d}u
\sim_{n\to+\infty}y^2_n-y^1_n\,,
\end{split}\]
where $y_n^j = \Gamma^{-1}(x_n^j)$. This implies for $n$ large enough, that $y^2_n=y^1_n$ which is a contradiction.
\end{proof}

\subsection{Hardy space on a curved strip}
Assume that there exists $f\colon \Omega\to S_\delta$ with $f', (f^{-1})'\in L^\infty(\Omega)$.
We are now in good position to define $\mathscr{H}^2(\Omega)$.
\begin{definition}\label{def.hardycurved}
We denote
\[
	\mathscr{H}^2(\Omega) = \{
		u\in \mathscr{O}(\Omega)\,, u \circ f\in \mathscr{H}^2(S_\delta)
	\}\,.
\]
\end{definition}
\begin{proposition}\label{prop:HardyCurved}
The following holds.
\begin{enumerate}[\rm (i)]
\item \label{it.traceCurved}
The trace operator $\mathscr{T}\colon v\in \mathscr{H}^2(\Omega)\mapsto [T(v\circ f^{-1})]\circ f\in L^2(\partial \Omega)$ is well-defined and moreover, $\mathscr{H}^2(\Omega)$ endowed with $\|\mathscr{T}\cdot\|_{L^2(\partial \Omega)}$ is a Hilbert space and $\mathscr{T}$ becomes an isometry. 
	\item\label{it.inclL2Curved} 
	We have the continuous embedding $\mathscr{H}^2(\Omega)\subset L^2(\Omega)\,,$ and for all $v\in \mathscr{H}^2(\Omega_\delta)$, 
\[
	\|v\|_{L^2(\Omega)}\leq \sqrt{\delta \|{(f^{-1})}'\|_{L^\infty}}\|
	\mathscr{T}v
	\|_{L^2(\partial \Omega)}\,.
\]
	\item \label{it.CauchyCurved} For $k\in\mathbb{N}$, $z_0\in\Omega$, $v\in \mathscr{H}^2(\Omega)$, we have
\[\begin{split}
	|v^{(k)}(z_0)|  
		&\leq
	\sqrt{\frac{(2k)!}{2^{2k+1}\pi}
	\frac{\|f'\|_{L^\infty}}{\|(f^{-1})'\|_{L^\infty}^{2k+1}}}
	{\rm dist}(z_0,\partial \Omega)^{-\frac{2k+1}{2}}
	\|
	\mathscr{T}v
	\|_{L^2(\partial \Omega)}\,.
\end{split}\]
\end{enumerate}
\end{proposition}

\begin{proof}
The proposition follows from Proposition \ref{prop:Hardy}. Let us develop some points. Let $v\in \mathscr{H}^2(\Omega)$ and define $u = (f^{-1})' v\circ f^{-1}\in \mathscr{H}^2(S_\delta)$. By Point \eqref{it.inclL2} of Proposition \ref{prop:Hardy}, we have
\[\begin{split}
	\|v\|_{L^2(\Omega)} &= \|u\|_{L^2(S_\delta)}\leq \sqrt{\delta }\|Tu\|_{L^2(\partial S_\delta)}
	\leq \sqrt{\delta \|(f^{-1})'\|_{L^\infty}}\|\mathscr{T}v\|_{L^2(\partial \Omega)}\,.
\end{split}\]
Point \eqref{it.inclL2Curved} follows.

Let $z_0,z_1\in\overline{\Omega}$, $\tilde z_0 = f(z_0)$, $\tilde z_1 = f(z_1)$, $\tilde \gamma\colon [0,1]\to S_\delta$ a smooth path such that $\tilde \gamma(0) = \tilde z_0$, $\tilde \gamma(1) = \tilde z_1$ and $\gamma = f^{-1}\circ \tilde \gamma$. We have
\[\begin{split}
	{\rm dist}_{\Omega}(z_0,z_1) \leq \int_0^1|\gamma'(t)|dt = \int_0^1|(f^{-1})'\circ\tilde \gamma(t)||\tilde\gamma'(t)|dt
	&\leq\|(f^{-1})'\|_{L^\infty}\int_0^1|\tilde\gamma'(t)|dt\,.
\end{split}\] 
Now, taking the infimum over all path in $\overline{S_\delta}$ between $\tilde z_0$ and $\tilde z_1$, we get
\[
{\rm dist}_{\Omega}(z_0,z_1)\leq \|(f^{-1})'\|_{L^\infty}{\rm dist}_{S_\delta}(\tilde z_0,\tilde z_1)\,,
\]
so that
\[
{\rm dist}_{\Omega}(z_0,\partial\Omega)\leq \|(f^{-1})'\|_{L^\infty}{\rm dist}_{S_\delta}(\tilde z_0,\partial S_\delta)\,,
\]
and Point \eqref{it.CauchyCurved} follows.
\end{proof}

We end this section by stating a useful density lemma, which follows from Lemma \ref{lem.density0}.

\begin{lemma}\label{lem.density}
The space $H^1(\Omega)\cap\mathscr{H}^2(\Omega)$ is dense in $\mathscr{H}^2(\Omega)$. More precisely, for all $u\in \mathscr{H}^2(\Omega)$, there exists $(u_n)\subset H^1(\Omega)\cap\mathscr{H}^2(\Omega)$ such that \[\lim_{n\to +\infty}\|\mathscr{T}(u_n-u)\|_{L^2(\partial\Omega)}=0\,.\]
\end{lemma}

\bibliographystyle{abbrv}
\bibliography{biblioguide}
\end{document}